\documentclass[letterpaper]{article} 
\usepackage{aaai24}  
\usepackage{times}  
\usepackage{helvet}  
\usepackage{courier}  
\usepackage[hyphens]{url}  
\usepackage{graphicx} 
\usepackage[numbers]{natbib}  
\usepackage{caption} 
\usepackage{algorithm}
\usepackage{algorithmic}

\usepackage{newfloat}
\usepackage{listings}
\DeclareCaptionStyle{ruled}{labelfont=normalfont,labelsep=colon,strut=off} 
\floatstyle{ruled}
\newfloat{listing}{tb}{lst}{}
\floatname{listing}{Listing}

\usepackage{url}            
\usepackage{booktabs}       
\usepackage{amsfonts}       
\usepackage{nicefrac}       
\usepackage{microtype}      
\usepackage{xcolor}

\usepackage{microtype,comment}

\usepackage{graphicx}

\graphicspath{ {./tree_plots/}{./images/} }

\usepackage{subfig}
\usepackage{booktabs}

\usepackage{amsmath}
\usepackage{amssymb}
\usepackage{mathtools}
\usepackage{amsthm}
\allowdisplaybreaks

\usepackage{mathrsfs}

\usepackage[capitalize,noabbrev]{cleveref}

\theoremstyle{plain}
\newtheorem{theorem}{Theorem}[section]
\newtheorem{proposition}[theorem]{Proposition}
\newtheorem{lemma}[theorem]{Lemma}

\newtheorem{condition}[theorem]{Condition}
\theoremstyle{definition}
\newtheorem{definition}[theorem]{Definition}
\newtheorem{assumption}[theorem]{Assumption}
\theoremstyle{remark}

\usepackage[textsize=tiny]{todonotes}

\usepackage{url}            
\usepackage{booktabs}       

\usepackage{nicefrac}       
\usepackage{microtype}      
\usepackage{xcolor}

\usepackage{microtype}
\usepackage{graphicx}
\usepackage{makecell}
\usepackage{caption}
\usepackage{multirow}

\usepackage{colortbl}
\definecolor{bgcolor}{rgb}{0.66,0.88,1.00}

\usepackage{tablefootnote}
\usepackage{threeparttable}
\usepackage{algorithm}
\usepackage{algorithmic}

\usepackage{xspace}

\xspaceaddexceptions{[]\{\}}

\usepackage{thm-restate}

\usepackage{eqparbox}
\renewcommand{\algorithmiccomment}[1]{\hfill\eqparbox{COMMENT}{// #1}}

\usepackage{amsmath,amsfonts,bm}

\def\eqref#1{equation~\ref{#1}}

\def\1{\bm{1}}

\DeclareMathAlphabet{\mathsfit}{\encodingdefault}{\sfdefault}{m}{sl}
\SetMathAlphabet{\mathsfit}{bold}{\encodingdefault}{\sfdefault}{bx}{n}

\def\gA{{\mathcal{A}}}

\def\gC{{\mathcal{C}}}

\def\gQ{{\mathcal{Q}}}

\definecolor{darkgreen}{rgb}{0,0.5,0}
\definecolor{darkred}{rgb}{0.7,0,0}
\definecolor{teal}{rgb}{0.3,0.8,0.8}
\definecolor{orange}{rgb}{1.0,0.5,0.0}
\definecolor{purple}{rgb}{0.8,0.0,0.8}
\definecolor{OliveGreen}{rgb}{0.7,0.7,0.3}

\newcommand{\rbra}[1]{\left( #1 \right)}
\newcommand{\sbra}[1]{\left[ #1 \right]}

\newcommand{\abs}[1]{\left| #1 \right|}

\usepackage{listings}
\usepackage{textcomp}
\usepackage{comment}

\title{Tree Search-Based Evolutionary Bandits for Protein Sequence Optimization}

\author{
    Jiahao Qiu\equalcontrib\textsuperscript{\rm 1},
    Hui Yuan\equalcontrib\textsuperscript{\rm 1},
    Jinghong Zhang\equalcontrib\textsuperscript{\rm 2},
    Wentao Chen\textsuperscript{\rm 3},
    Huazheng Wang\textsuperscript{\rm 4},
    Mengdi Wang\textsuperscript{\rm 1}
}
\affiliations{
    
    \textsuperscript{\rm 1}Princeton University\\
    \textsuperscript{\rm 2}University of California San Diego\\
    \textsuperscript{\rm 3}MLAB Biosciences Inc\\
    \textsuperscript{\rm 4}Oregon State University\\
    jq3984@princeton.edu, huiyuan@princeton.edu, zhangjinghong99@outlook.com, wentao.chen@mlabbiosciences.com, huazheng.wang@oregonstate.edu, mengdiw@princeton.edu

}

\begin{document}

\maketitle

\begin{abstract}
  While modern biotechnologies allow synthesizing new proteins and function measurements at scale, efficiently exploring a protein sequence space and engineering it remains a daunting task due to the vast sequence space of any given protein. Protein engineering is typically conducted through an iterative process of adding mutations to the wild-type or lead sequences, recombination of mutations, and running new rounds of screening. To enhance the efficiency of such a process, we propose a tree search-based bandit learning method, which expands a tree starting from the initial sequence with the guidance of a bandit machine learning model. Under simplified assumptions and a Gaussian Process prior, we provide theoretical analysis and a Bayesian regret bound, demonstrating that the combination of local search and bandit learning method can efficiently discover a near-optimal design. The full algorithm is compatible with a suite of randomized tree search heuristics, machine learning models, pre-trained embeddings, and bandit techniques. We test various instances of the algorithm across benchmark protein datasets using simulated screens. Experiment results demonstrate that the algorithm is both sample-efficient and able to find top designs using reasonably small mutation counts.
 
\end{abstract}

\section{Introduction}



Advances in biotechnology have demonstrated human's unprecedented capabilities to engineer proteins. They make it possible to directly design the amino acid sequences that encode proteins for desired functions, towards improving biochemical or enzymatic properties such as stability, binding affinity, or catalytic activity. 
Directed evolution (DE), for example, is a method for exploring new protein designs with properties of interest and maximal utility, by mimicking the natural evolution process. The development of DE was honored in 2018 with the awarding of the Nobel Prize in Chemistry to Frances Arnold for the directed evolution of enzymes, and George Smith and Gregory Winter for the development of phage display \citep{arnold1998design, smith1997phage, winter1994making}. Traditional DE strategies are inherently screening (greedy search) strategies with limited ability to generate high-quality data for probing the full sequence-function relationships. Recent advances in synthetic DNA generation and recombinant protein production make the measurement of protein sequence-function relationships reasonably scalable and high-throughput \citep{packer2015methods, yang2019machine}.



\begin{figure*}[h!]
 \vspace{-3mm}
  \centering \includegraphics[height=0.22\textwidth, width=0.7\textwidth]{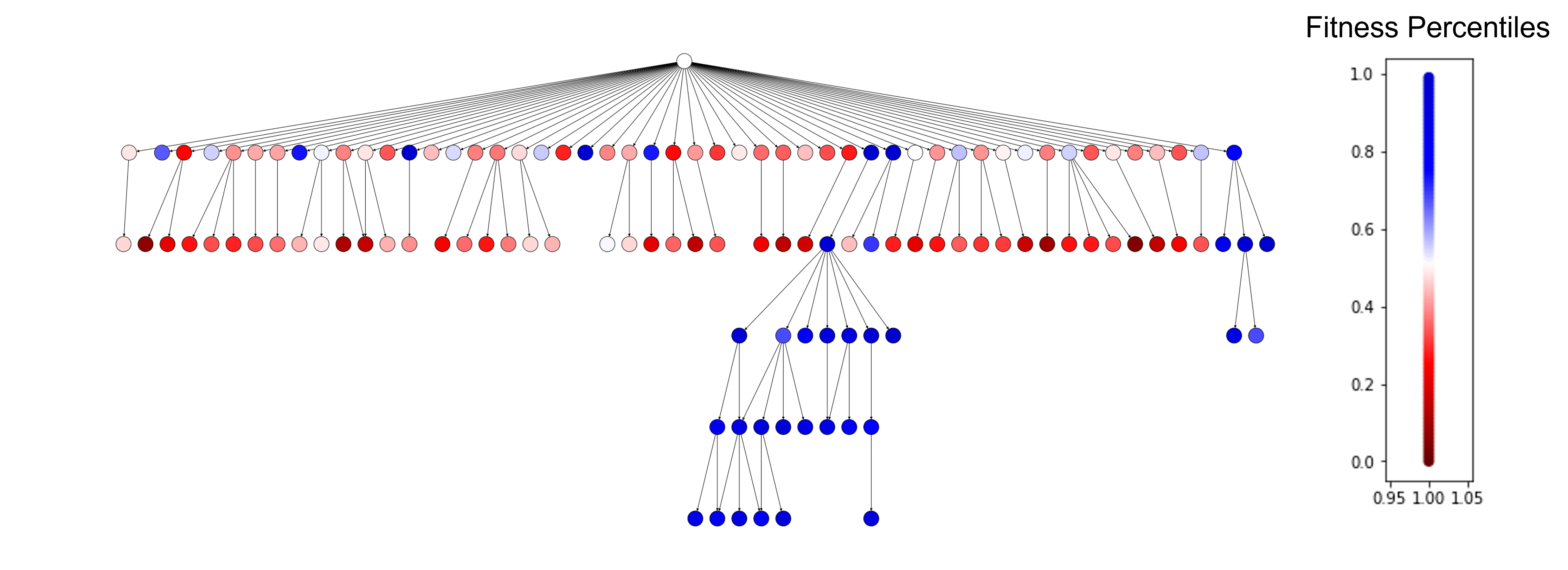}
  \caption{A Tree visualization of AAV screen dataset \citep{Bryant2021}, generated by starting from the wild-type and building the tree via downsampling children with an editing distance of 1 from the parent. The wet-lab screen initiates with a wild-type design sequence (root node), and in each round new sequences are generated by adding randomization and keeping those with high fitness scores as parents. It was believed that nodes with high fitness are more likely to generate high-fitness children.  \label{tree-data}}
  \vspace{-6mm}
\end{figure*}





Due to the bottleneck of wet-lab experimentation and the complex landscape of protein functions, identifying novel protein designs for maximal fitness remains one of the most difficult but high-value problems in modern medicine and biology. 
This has motivated scientists to apply machine learning approaches, beginning with \citet{fox2003optimizing} and followed by many, with increasing amounts of efforts utilizing \textit{in silico} exploration and machine learning beyond experimental approaches \citep{yang2019machine, fannjiang2020autofocused, doppa2021adaptive, shin2021protein, freschlin2022machine,wang2022neural,sinai2020adalead}. More recent advances in large language models open up new opportunities for modeling and predicting protein functions and generalizing knowledge across protein domains \citep{rives2021biological, Shuai2021.12.13.472419, nijkamp2022progen2, Hsu2022Learning, Elnaggar2020ProtTrans}. 


The key research challenge with designing the iterative protein screening strategy is {\it exploration}, i.e., how to effectively explore in a large combinatorial space and learns the sequence-to-function landscape towards finding the optimal. While many attempts have been proved successful in simulation and sometimes in real experiments \citep{Bryant2021, shin2021protein}, they often are limited by practical constraints and their performance is very sensitive to domain/distribution shifts. Even with the best and largest pre-trained protein language models such as ESM-1b \cite{rives2019biological} and ProGen2 \cite{nijkamp2022progen2}, one often needs to explore an almost unknown domain and learn a new function map in order to discover new drugs. This is especially true with antibody engineering. Antibodies have highly diverse complementarity-determining region (CDR) sequences that can be altered, resulting in a huge sequence space to explore for optimal properties. The binding of antibodies to their targets are extrinsic properties of antibodies and it is difficult to accurately model the sequence-binding relationships solely from the sequences alone. Further, most of the exploration strategies used in practice lack theoretical guarantees.



\paragraph{Practical considerations in protein screens and a tree search view}

Protein engineering is typically done through trial-and-error approaches in altering the primary sequence by mutations or changing the length of certain regions. More high-throughput approaches (e.g. {\it in vitro} display systems) involve randomizing the amino acids for the positions-of-interest or region-of-interest. To obtain the optimal properties for a protein, the directed evolution approach is typically used by exploring a limited sequence space in each round of engineering, and starting the next round of engineering from the best one or the best few sequences in the previous round, in an iterative manner \cite{wu2019machine}. A typical protein engineering workflow by producing purified proteins is limited to up to a few hundred sequences due to the throughput limitations. {\it In vitro} display systems, on the other hand, can be used to screen millions of sequences, although the data generation (labeled data of sequence-function relationships) throughput is much smaller. 

For practical reasons, especially for therapeutic proteins, the choice of mutations is dependent on the knowhows of the protein engineer, and the number of mutations is kept low in order to prevent unexpected issues associated with decreased protein stability, compromised binding specificity, and immunogenicity. Adding too many mutations may also lead to distribution-shift and reduce the robustness of machine-learning models.  Thus, practitioners are reluctant to make large jumps in the screening/search process.

For example, \citep{Bryant2021} studied the engineering of Adeno-associated virus 2 capsid protein (AAV) and screened a total of 201,426 variants of the AAV2 wild-type (WT) sequence. It screened all single mutations in the first round, then generated variants with $>1$ mutations via randomization and selection of high-value ones in later rounds. 
Such an iterative process mimics a tree search. To understand this process, we visualize those sequences from the dataset of \citep{Bryant2021} after downsampling in Figure \ref{tree-data}. We see that the screen data map nicely to a tree, where the root node corresponds to the wide-type AAV and mutations connect parent and child nodes. 
These observations are consistent with practical screening strategies that add mutation sequentially and search for better alternatives.
Note that the tree size grows exponentially as mutations are added. For the example of AAV, variants with up to $5$ mutations form a tree with $\sum_{i = 0}^5 \binom{28}{i} \cdot 19^i$ nodes. Thus even with a bounded number of mutations, the problem is prohibitively difficult.


\paragraph{Our approach} 
In order to make the screening process more efficient, we borrow ideas from both the protein engineering practices and recent advances in bandit machine learning, hoping to get the best from both worlds. We follow two principles for algorithm design: 
\begin{enumerate}
\item[(1)] We wish to largely follow a tree search process and identify optimal sequences with just a few mutations. As mentioned, practitioners are reluctant to make large jumps in the screening/search process.  Being a \emph{local search} strategy, tree search from lead sequences will keep total mutation counts small, which means better reliability of the found solution. Further, machine learning models for function prediction are more likely to generalize well to new designs that do not change too much from training data. 

\item[(2)] We employ bandit exploration techniques to guide the tree branching process. Instead of searching greedily using a learned prediction model, we hope to more aggressively search designs with higher uncertainty. Two main techniques in bandit learning are upper confidence bound (UCB) and posterior sampling (also known as Thompson Sampling, aka TS). We will incorporate these bandit techniques, leveraging pre-trained protein sequence embedding and neural networks, into tree search to enhance exploration.
\end{enumerate}

In this paper, we propose to combine tree search with bandit machine learning. We begin by presenting a meta-algorithm (Algorithm \ref{alg:meta}). It proceeds by mimicking the directed evolution process, growing a tree from the root node, and gradually expanding via mutation and recombination. It uses a pre-trained embedding and a machine learning model for predicting fitness, and during the search process, it adopts a bandit strategy to update the predictor and actively explore the tree. This meta-algorithm provides a versatile framework for analyzing exploration in sequence space.

\paragraph{Results}
For theoretical analysis, we study a Bayesian setting where the true function map has a Gaussian Process prior distribution. We also assume Lipschitz continuity of the embedding map and local convexity of the fitness function. Under these simplified conditions, we show that the meta-algorithm with GP bandit can provably identify the optimal sequence and achieves a regret $O \rbra{\gamma_T \sqrt{T}}$, where $\gamma_T$ is known as the maximal information gain. The theoretical analysis may apply to a broader class of bandit algorithms and be of independent interest.

Next, we fully develop the algorithm for numerical implementation, and make it compatible with a suite of bandit models including UCB and Thompson Sampling. We experiment with instances of the algorithm and compare with a variety of baselines, using simulation oracles trained from real-life protein function datasets AAV \cite{Bryant2021}, TEM \cite{Gonzalez2019-fr} and AAYL49 antibody \cite{Engelhart2022} datasets. Experiments results show that tree-based methods achieve top performances across benchmarks and can efficiently find near-optimal designs with single-digit mutation counts. 





\section{Related Work}\label{sec:related}

\paragraph{Protein engineering.}

The traditional DE works by artificially evolving a population of variants, via mutation and recombination, while constantly selecting high-potential variants \citep{chen1991enzyme, chen1993tuning, kuchner1997directed, hibbert2005directed, turner2009directed, packer2015methods}.
Many variations of DE methods allow targeted randomization of positions-of-interest or regions-of-interest.
It is also possible to synthesize specific variants and operate on the combinatorial space likewise with high-throughput method, at a  cost, and this allows directly applying a Gaussian process bandit algorithm \cite{romero2013navigating}. A preliminary verision of this paper appeared at \citep{wang2022neural}. See \cite{yang2019machine} for a high-level survey of machine learning-assisted protein engineering, and see \cite{fox2007improving, bedbrook2017machine} for more examples.

\paragraph{Search algorithms for protein sequence design.}
Researchers have tried out various machine learning-based methods for sequence optimization. Bayesian Optimization(BO)\cite{Mockus1989BayesianAT} is a classical method for optimizing protein sequences. We use the code developed by \citet{sinai2020adalead} who uses an ensemble of models for BO as one of the baselines. LaMBO\citep{pmlr-v162-stanton22a} is also a BO-based algorithm that supports both single-objective and multi-objective. Design by adaptive sampling (DbAS; \citet{brookes2018design}) and conditioning by adaptive sampling (CbAS; \citet{brookes2019conditioning}) use a probabilistic framework. DyNA PPO\cite{Angermueller2020Model-based} uses proximal policy optimization for sequence design. PEX MufacNet\cite{ren2022proximal} is a local search method based on the principle of proximal optimization. We discussed more about the search algorithms in Appendix \ref{sec:related_work_details}.

\paragraph{Bandit learning.}
Bandit is a well-studied framework for optimizing with uncertainty, powerful in balancing exploration-exploitation trade-off -- a core challenge in protein sequence optimization. Therefore, it is potential to study protein optimization from a bandit perspective and there is recent work \cite{yuan2022bandit} emerging from this middle ground.
To balance exploration and exploitation, typical strategies are being optimistic in the face of uncertainty (OFU) based on the upper confidence bound (UCB) \citet{abbasi2011improved} and Thompson Sampling (TS) \citet{russo2014learning}, which randomizes policies based on the posterior of the optimal policy.
Regret guarantees have been established for different function classes of rewards, starting from the line of linear bandits. It was shown for the $d$-dim linear model upper and lower regret bounds meet at  $\tilde O \left( d \sqrt{T} \right)$\citep{Auer02,li2010contextual,abbasi2011improved, li2010contextual, russo2014learning}, with extensions to reinforcement learning \citep{yang2020reinforcement,yang2019sample,li2022communication} and sparsity-aware or decentralized settings \citep{hao2020high,li2022communication}. 
Later on, results from linear bandits have been extended to kernelized/ Gaussian process (GP) bandits. With $\gamma_T$ defined as the maximal information gain and $d$ being the dim of actions, \cite{srinivas2009gaussian} showed an upper bound of $\tilde O \rbra{\sqrt{dT\gamma_T}}$ and one of $\tilde O \rbra{\gamma_T \sqrt{T}}$ in the agnostic setting achieved by GP-UCB. Exploded by a factor of $\sqrt{d}$, $\tilde O \rbra{\gamma_T \sqrt{dT}}$ regret was shown by \cite{chowdhury2017kernelized} for agnostic GP-TS.

Recently, there has been a growing interest in bridging the predictive power of deep neural networks with the exploration mechanism of bandit learning, which is known as neural bandits \citep{zhou2020neural,zhang2020neural,jacot2018neural,pan2022neuralbandits}. Building upon the neural
tangent kernel (NTK) technique to analyze deep neural network \cite{jacot2018neural}, NeuralUCB \citep{zhou2020neural} can be viewed as a
direct extension of kernel bandits \cite{srinivas2009gaussian}. \citet{zhang2020neural} proposed the Thompson Sampling version of neural bandits. \citet{pan2022neuralbandits} proposed NeuralLinUCB, which learns a deep representation to transform the raw feature vectors and
performs UCB-type exploration in the last linear layer of the neural network.

\section{Method}\label{sec:method}

\subsection{Problem Formulation}

Let $x\in\mathcal{X}$ denote an amino-acid sequence, and let $F(x)$ denote a function measuring the fitness of the sequence. In practice, a known embedding map $\phi: \mathcal{X} \to \mathbb{R}^{d}$, mapping a sequence $x$ to its embedding vector $\phi(x)$, is often utilized to model the fitness $F(x)$ as $f^{\star}(\phi(x))$. 
Thus, we formulate the sequence design problem as
\begin{equation}
\label{seq_obj_w_emb}
    \max_{x\in\mathcal{X}} F(x) := f^{\star}(\phi(x)),
\end{equation}
where $f^{\star}$ represents the unknown ground-truth function. 
Here $\mathcal{X}$ corresponds to the tree rooted at the $x^{wt}\in \mathcal{X}$ of a fixed depth.

The learning problem is to explore $\mathcal{X}$, conduct screens and collect data points of the form $(\phi(x), \tilde F(x))$, refine estimates of $f^{\star}$ in an iterative fashion. Here $\tilde F(x)$ represents a noisy measurement of the unknown true fitness $F(x)$.

The hardness of the problem is due to searching over the combinatorial space $\mathcal{X}$. Although we use an embedding map $\phi(x)$ to help learn $F$ more accurately, we still have to work with discrete sequences and cannot make jumps in the embedding space. This nature of discrete sequence optimization is in sharp contrast to typical bandit settings.

\subsection{Meta Algorithm}
We present a meta-algorithm that combines bandit machine learning with local tree search in Algorithm \ref{alg:meta}. Detailed implementation of the subroutine $\textsc{TreeSearch}$ (Alg. \ref{alg:cap} ) is delayed till Section \ref{sec:full_alg} and \ref{sec:exp}. 

\vspace*{-3mm}
\begin{algorithm}[H]

\caption{Meta algorithm}

\label{alg:meta}
\begin{algorithmic}[1]
\STATE \label{line:tree_search} \textbf{Input: }  wildtype $x^{wt}$, total rounds $T$
\STATE \textbf{Initialization: } Add $x^{wt}$ to active node set $\gA_0$ as root node. 
\FOR{$t=1,2,... T$}
        \STATE Update bandit model:
        return a scoring function $F_t(\cdot) := f_t (\phi(\cdot))$.
        \STATE   Tree search by Alg. \ref{alg:cap}: $\gA_t \leftarrow \textsc{TreeSearch}(\gA_{t-1}, F_t)$.

        \STATE Query the fitness $F(x)$ for some (or all) $x$'s in $\gA_t$ and append $\{(\phi(x), \tilde F(x))\}$ to the dataset.
        
\ENDFOR
\end{algorithmic}
\end{algorithm}
\vspace*{-4mm}

\begin{figure*}[!ht]
    \vspace{-2mm}
  \centering 
  \includegraphics[height=0.20\textwidth, width=0.9\textwidth]{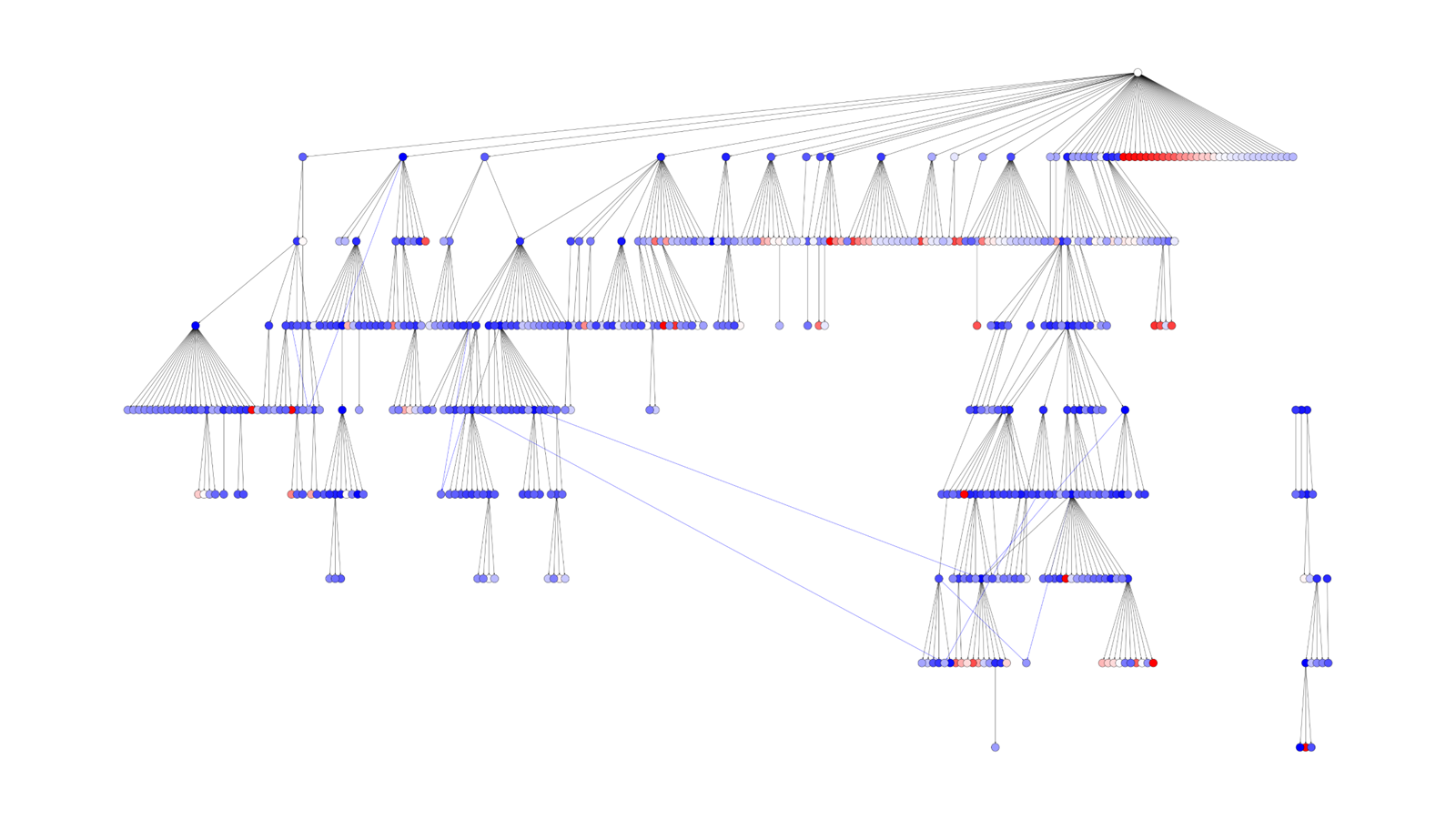}
  \vspace{-2mm}
  \caption{Visualization of tree search process of Algorithm \ref{alg:meta} using the AAV oracle. The search initiates with a wild-type sequence (the root note) and in each round, we choose 100 sequences generated from the last round according to scores derived from UCB and TS by single mutation which leads to a node in the next layer and recombination of sequences (shown by blue edges) which leads to a jump to a layer with more mutations. A path to the optimal sequence is shown by the bold line.}\label{fig:tree_visualization}
  \vspace{-5mm}
\end{figure*}

\subsubsection{Analysis of simple tree search}

Motivated by the practical observation that 
high-potential variants usually appear within a small count of mutations from a wildtype and the practical consideration that function approximation generalizes better in a local region where most data lies, we focus on a local searching region $\bar{\mathcal{X}} := \{x: d(x, x^{wt}) \leq N\} \subset \mathcal{X}$ and aim to identify the optimal sequences within $\bar{\mathcal{X}}$. Here $d(\cdot,\cdot)$ denotes the hamming distance between any two arbitrary sequences in $\mathcal{X}$ and $N$ controls the mutation counts (the depth of tree search). 

\paragraph{Bayesian regret}
Given the search region $\bar{\mathcal{X}}$, we measure the performance of an algorithm by its Bayesian regret defined as follows.
\begin{definition}
    \label{def:bayes_rgt}
    Suppose the algorithm finds a series of protein sequences $\{x_{t}\}_{t=1}^{T}$ within $\bar{\mathcal{X}}$, then its accumulated Bayesian regret is defined as
    \begin{equation}
        \operatorname{BayesRGT}(T) = \mathbb{E}\left[\sum_{t=1}^T \rbra{ \max_{x \in \bar{\mathcal{X}} } F(x) - F(x_{t})} \right],
    \end{equation}
    where $\mathbb{E}$ is taken over the prior distribution of $f^{\star}$ (which we will assume to be a Gaussian process in Section \ref{subsec:gp_theory}), and taken over other randomness in finding $\{x_{t}\}_{t=1}^{T}$ .
\end{definition}

\paragraph{Simple tree search}
In order to obtain a basic theoretical understanding of the tree search-based bandit learning process,
we consider a {\it simplified} mathematical abstraction of Alg. \ref{alg:meta}: suppose in each round $t$, protein $x_t$ is filtered out through the subroutine $\textsc{TreeSearch}(\gA_{t-1}, F_t)$, and $\{x_t\}_{t=1}^{T}$ satisfies the following condition. And it only collects $\{(\phi(x_t), \tilde F(x_t))\}$ in the active set for updating $f_t$.

\begin{condition} 
\label{cond:local_argmax}There exists $r>0$ such that each iteration of tree search is able to find a sequence that has equal or better $F_t$ value than the local maximum within radius $r$ in the embedding space, i.e.,
\begin{equation}
    F_t( x_{t} ) \geq \max \{ F_t(x) \mid x: \|\phi(x)-\phi(x_{t-1})\| \leq r\}.
\end{equation}
\end{condition}

This condition shows that the local search method makes sufficient improvement w.r.t. $F_t$ per iteration. In practice, this condition can be satisfied by tuning parameters and stopping conditions of the tree search subroutine.
To understand this, when $\phi$ is Lipschitz, with notation $\phi(\bar{\mathcal{X}}): = \{\phi(x): x \in \bar{\mathcal{X}}\}$, we have
$\phi(\bar{\mathcal{X}}) \subset \mathcal{D}:= \{\phi(x): \| \phi(x) - \phi(x^{wt})\| \leq R:= L_{\phi} N\}$, $\mathcal{D}$ is the local region around $\phi(x^{wt})$ in the embedding space.
Thus in practice, if a good embedding map $\phi$ successfully captures the latent structure of $\mathcal{X}$ such that $\phi(\bar{\mathcal{X}})$ spans well in $\mathcal{D}$, then Condition \ref{cond:local_argmax} is satisfied whenever $\textsc{TreeSearch}$ sufficiently exhausts the local variants around $x_{t-1}$ in the discrete sequence space.
In other words, properties of the embedding map $\phi$ critically connect locally searching in the discrete sequences to searching in the embedding space towards improving function value.

\begin{figure}[ht]
    \vspace{-2.5mm}
  \centering
    \includegraphics[height = 0.4\linewidth]{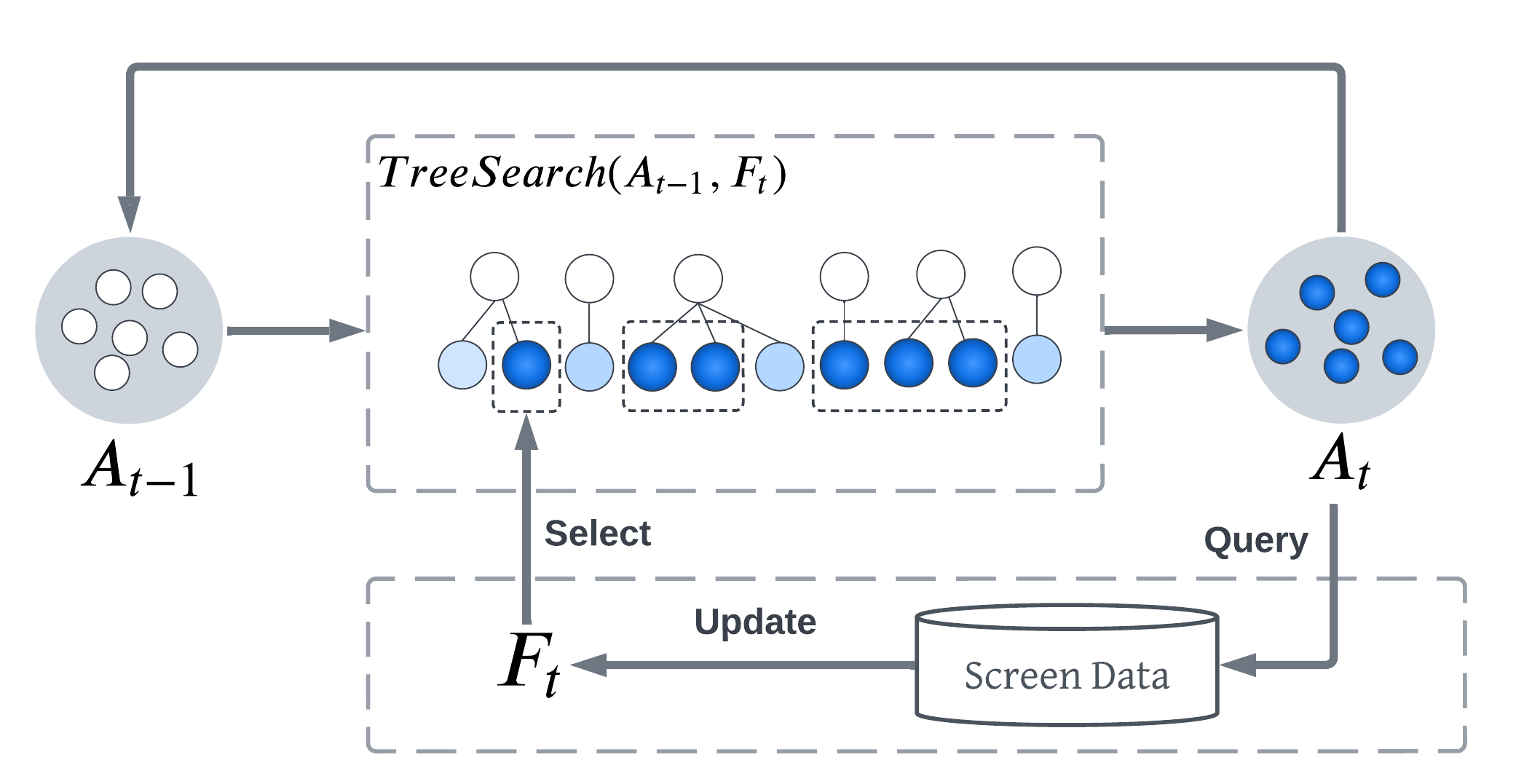} 
    \caption{Diagram of the meta algorithm (Alg.\ \ref{alg:meta})}\label{fig:metaAlg}
    \vspace{-2.5mm}
\end{figure}

\subsection{Regret Theory under GP Fitness}
\label{subsec:gp_theory}
In this section, we provide a Bayesian regret theory for the proposed meta algorithm, under some necessary assumptions and simplifications. We assume a GP prior-posterior modelling of $f^{\star}$ and $f_t$'s, as well as some local properties of $f_t$'s. The basic theoretical understanding of the tree search-based bandit learning process we obtained in this section may have implications for the broader settings beyond our assumptions/simplifications. 

\subsubsection{GP Modeling of $f^{\star}$}

For theoretical analysis, we consider a Bayesian learning setting where the ground truth $f^{\star}$ is assumed to follow a Gaussian Process (GP) prior, and the evaluation of $F(x)$ is assumed to be corrupted by Gaussian noise. GP model is well studied and widely applied in machine learning \cite{seeger2004gaussian}, especially in classic bandit optimization \cite{srinivas2009gaussian, chowdhury2017kernelized}.

\begin{assumption}
\label{asmp:gp_f}
$f^{\star}:\mathbb{R}^d \to R$ is a sample from the Gaussian process $\mathcal{GP}(0, k(\cdot, \cdot))$ (with known kernel function $k(\cdot, \cdot)$) as priori, i.e.
\begin{equation}
    f^{\star}(\phi(x)) \sim \mathcal{GP} \rbra{0, k(\phi(x), \phi(x^{\prime}) )}.
\end{equation}
\end{assumption}

\begin{assumption}
\label{asmp:noisy_fdb}
For any sequence $x$, querying its fitness $F(x)$ returns a noisy feedback $\tilde F(x)$ corrupted by a Gaussian, i.e.
\begin{equation}
     \tilde F(x) = f^{\star}(\phi(x)) + \epsilon, 
\end{equation}
where $\epsilon \sim \mathcal{N} \rbra{0, \lambda}, \lambda > 1$ and is sampled independently from the history.
\end{assumption}

To get a sharper regret bound, we sample $f_t$ from the posterior Gaussian Process with monitoring. Please refer to Appendix \ref{subsec:rare_swt} for more details on the Gaussian Process and the monitored updating rule of $f_t$. For the subsequent theory, we also need some assumptions on the local properties of $f_t$'s. Recall $\mathcal{D}:= \{\phi(x): \| \phi(x) - \phi(x^{wt})\| \leq R:= L_{\phi} N\}$ is the local region surrounding $\phi(x^{wt})$ in $\mathbb{R}^{d}$.
\begin{assumption}
    \label{asmp:ft_cvx}
    At each time step $t$, $f_t$ is bounded in $\mathcal{D}$. For $\forall t, z \in \mathcal{D},$
    \begin{equation}
        |f_t(z)| \leq B.
    \end{equation}
    Also, assume for all $t$, $f_t$ is locally concave, $L_f$ Lipschitz in $\mathcal{D}$ and attains its maximal value in the interior of $\mathcal{D}$.
\end{assumption}

Next in Theorem \ref{thm:main}, Alg. \ref{alg:meta} is proven to explore the local sequence space via tree search and bandit learning, while attaining low regret. For clarity and rigorousness, we provide a rewritten version of Alg. \ref{alg:meta} to reflect all simplification and technical adjustments we made, see Alg. \ref{alg:evo_ker_mnt_TS} in Appendix \ref{appx_proof}.

\begin{theorem}
\label{thm:main}
Under Assumption \ref{asmp:gp_f}, \ref{asmp:noisy_fdb}, 
\ref{asmp:ft_cvx} and Condition \ref{cond:local_argmax}, Alg.\ref{alg:meta} updates $f_t$ for $O \rbra{\gamma_T}$ times its Bayesian regret is bounded by
\begin{equation}
\label{rgt_bound}
        \operatorname{BayesRGT}(T) = O \rbra{  \beta_{T}\sqrt{\lambda T \gamma_{T}} + B \gamma_{T} \rbra{1 + \frac{4 L_{\phi}^2 N^2}{r^2}}},
   \end{equation}
   where $\beta_T = O\left( \mathbb{E} \sbra{\|f^{\star}\|_{k}}+\sqrt{d \ln T} + \sqrt{\gamma_{T-1}}\right)$. $\gamma_T$ is the information gain,
   $r$ is inherit from Condition \ref{cond:local_argmax} and $N$ is the depth of tree search.
\end{theorem}

\paragraph{Rate of regret}  
The highest order term in (\ref{rgt_bound}) is of $\Tilde{O} \rbra{\gamma_T \sqrt{T} \vee  \sqrt{d \gamma_T T}}$, which matches the results by \cite{srinivas2009gaussian, chowdhury2017kernelized} studying GP bandits.  (\ref{rgt_bound}) turns out to be $\tilde O\rbra{d \sqrt{T}}$ (with $\gamma_T = d \log T$) when $k(\cdot,\cdot)$ in the prior is linear, which downgrades the GP model to the $d$-dim Bayesian linear model. It's worth mentioning that the recovered $\tilde O\rbra{d \sqrt{T}}$ bound improves \cite{yuan2022bandit} by a factor of $\sqrt{d}$, benefiting from the stability brought by $f_t$'s rare switching.
Here $\gamma_T$ is closely related to the ``effective dimension" of the chosen kernel, which is a natural and common complexity metric for online exploration.

\paragraph{Novelty and significance compared to classical results of bandits and Bayesian optimization.}
We highlight that classical results {\it do not} apply our tree search algorithm that uses local search to gradually explore the action space. In particular, classical methods (such as plain vanilla UCB or TS) require finding the maximum of a surrogate function within the search region. This is far from the protein design practice where solutions with large surrogate values are approached by search: looking for the argmax of $\hat f_t$ can lead to instability and invalid solution. 

In contrast, our method adds mutations gradually in a tree search process mimicking real-world screen practice. Analysis of such methods is much more complicated: We adapt classic arguments in optimization theory to analyze the progression of local search, which is further entangled with bandit exploration and information gain when new samples are collected. The proof idea is to view each local update as a form of proximal point optimization update and derive a novel recursive decomposition of the overall regret. Please see Appendix \ref{appx_proof} for full proof details. Our analysis may be of interest to analyzing a broader class of bandit-guided evolutionary optimization algorithms. 

Despite these differences, our regret bound nearly matches regret bound of classical bandit methods. The message is quite positive: The use of iterative local search does not impair the quality of learning, with provable guarantee to explore the space of interests. 

\paragraph{Universality of GP assumption}

GP provides a universal function approximation and it comes with both practical and theoretical implications.  GP model and Bayesian optimization have been directly applied in protein engineering practice and proved effective in wet-lab experiments \cite{romero2013navigating},  which supports our assumptions here. Further, GP can represent any kernel function space, which together with the modern deep learning theory \cite{jacot2018neural} imply a powerful theoretical approximation to neural networks used in practice.

\paragraph{Synergy between Tree search and Bandit} 

Our analysis demonstrates how to create synergy between tree search and bandit machine learning. If we just use pure tree search without any bandit learning, tree search has a complexity that grows exponentially with depth. With the use of bandit learning, the sample complexity reduces to quadratic with tree depth.

\section{Full Algorithm}
In this section, we present the full algorithm.

\begin{figure}
\vspace{-3mm}
  \centering
    \includegraphics[height = 0.4\linewidth]{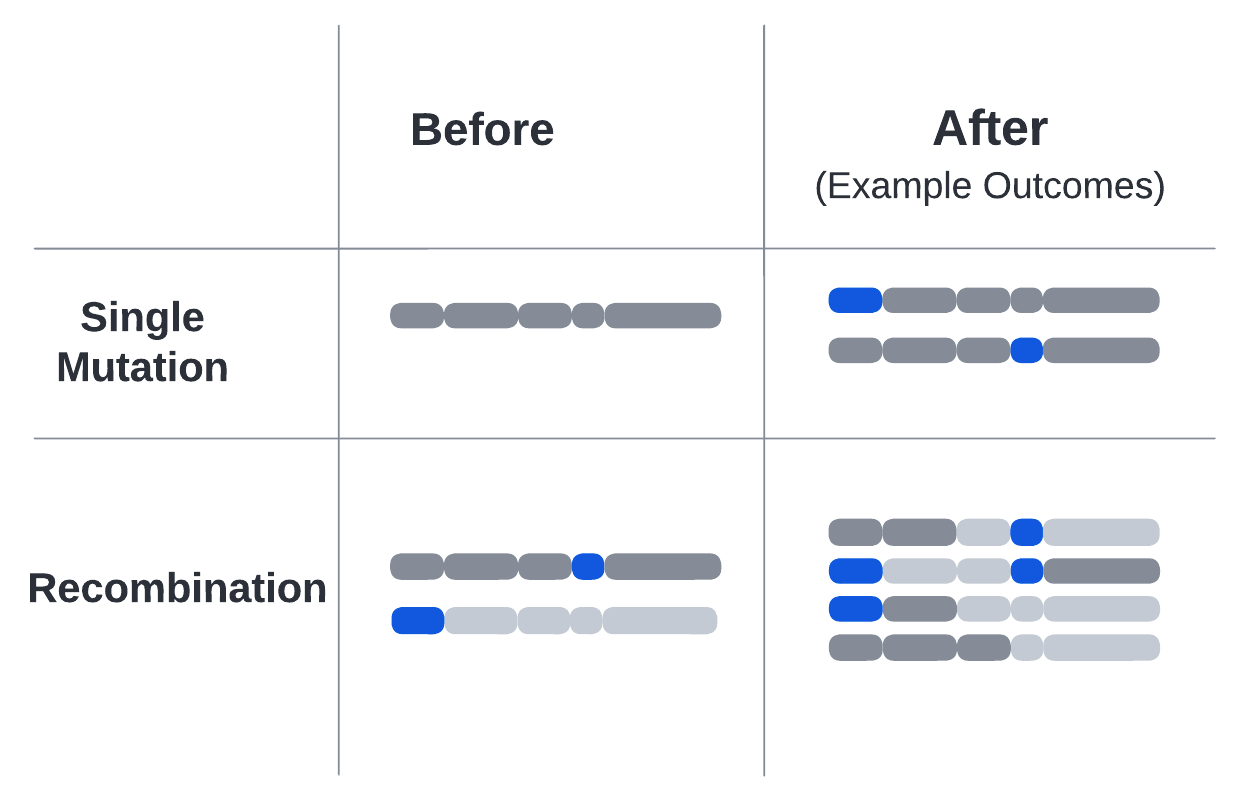} 
    \vspace{-2mm}
    \caption{A demonstration of mutation and recombination. \label{fig:sigmut_rcb}}
    \vspace{-7mm}
\end{figure}

\subsection{Tree Search Heuristics}

Algorithm \ref{alg:cap} expands a tree starting from the root wide-type sequence, i.e., $\mathcal{A} = \{x^{wt} \}$. Similar to the practice of directed evolution, we generate child nodes from parents in two ways: adding random mutation to one sequence and pairwise recombining two sequences. See Figure \ref{fig:sigmut_rcb} for an illustration of these two operations. We use hyperparameters $n_1,n_2$ to control the rate of mutation and rate of recombination.

We use an active set $\mathcal{A}$ to keep track of the frontier of the tree search. At each round, the tree expands in a randomized way. New nodes with high scores measured by $F$ will be kept track of and later used for updating the active set $\mathcal{A}$.  

Ideally, one would want to keep the full search history in $\mathcal{A}$, but then the runtime will quickly blow up due to the exponential tree size.
To make it more computationally affordable, we do not expand the active set $\mathcal{A}_t$, but keep it at a constant size in our implementation. We use a parameter $\rho$ to control the portion of previously visited nodes to keep in the active set. This parameter can also be viewed as balancing depth and width in tree search.

\vspace*{-2mm}

\begin{algorithm}[H]
\caption{\textsc{TreeSearch}($\gA, F(\cdot)$)}
\label{alg:cap}
\begin{algorithmic}[1]
\STATE \textbf{Input: }  Active node set $\gA$, scoring function $F(\cdot)$
\STATE \textbf{Parameter: } $n_1, n_2, n_3, \rho$
\STATE \textbf{Initialization: } Candidate node set $\gC = \emptyset$, Query node set $\gQ = \emptyset$
        \FOR {$i=1,2,\cdots,n_1$}
            \STATE Add random mutation to a random sequence $x \in \gA$.\label{line:single_mutation}
            \STATE Add the new sequence to candidate node set $\gC$.
        \ENDFOR
        \FOR {$i=1,2, \cdots, n_2$}
            \STATE Sample $x, y$ from $\gA$ uniformly with replacement.
            \STATE Recombine $x,y$ and add the new sequence to $\gC$.
        \ENDFOR
        
        \STATE Set new active node set $\gA$ with $n_3$ sequence where $\rho$ portion is from $\gA$ and $1-\rho$ portion is from the top  scored sequences in $\{F(x): x\in \gC\}$. 
        
        \STATE \textbf{return } $\gA$

\end{algorithmic}
\end{algorithm}

\vspace*{-2mm}

\label{sec:full_alg}

\subsection{Bandit Exploration}
We use two classic exploration methods for scoring, Upper Confidence Bound (UCB) \citep{li2010contextual} and Thompson Sampling (TS; \citet{russo2014learning}), to enable our algorithm to consider the potential of each node and look ahead. Details are referred to Appendix \ref{subsec:ban}.

\begin{figure*}[t]
    \centering
    \includegraphics[height = 0.23\textwidth, width=0.85\textwidth]{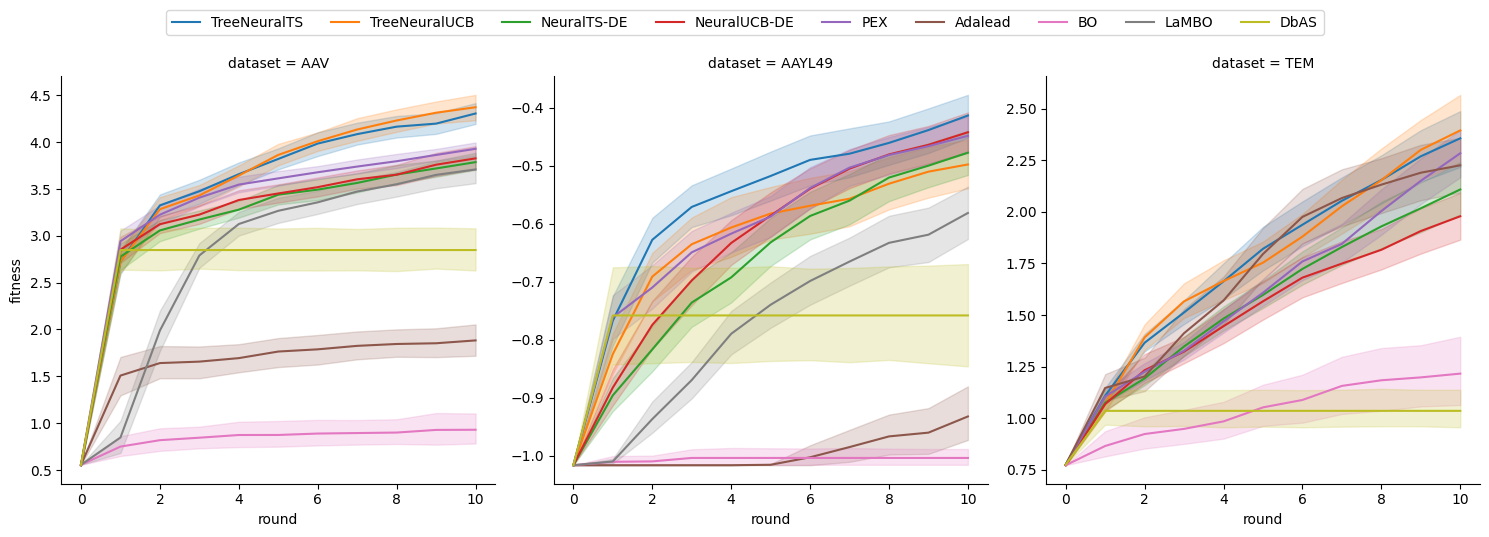}
    \includegraphics[height = 0.19\textwidth, width=0.85\textwidth]{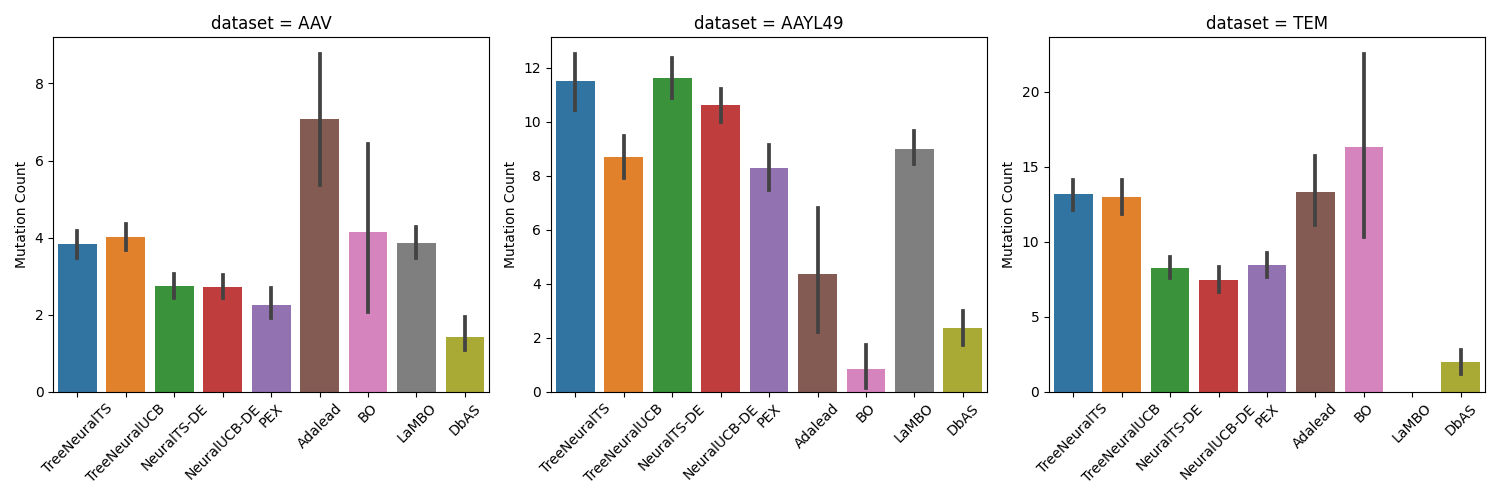}
    \caption{Learning curves of algorithms with comparison to baselines, tested over three datasets.}\label{fig:comparison_baseline}
    \vspace{-6mm}
\end{figure*}

\section{Experiment}\label{sec:exp}

In this section, we evaluate our algorithm using oracles simulated to mimic the exploration process in wet-lab protein screens. Following \citet{ren2022proximal} and \citet{sinai2020adalead}, we use oracles to mimic the fitness landscape as the wet-lab experiment is both time-consuming and cost-intensive. We build experiments around real-world protein datasets, such as AAV \cite{Bryant2021}, TEM \cite{Gonzalez2019-fr} and AAYL49 antibody \cite{Engelhart2022}. Experiments results show that the tree search-based bandit outperforms existing baselines and leads to several interesting observations.

\subsection{Datasets and Setup}

\subsubsection{Datasets}

We experiment using three datasets from protein engineering studies and train oracles to simulate the ground-truth wet-lab fitness scores $f^{\star}$ of the landscape. The AAV and TEM oracles use pre-trained TAPE embedding \cite{tape2019} with a CNN model and the AAYL49 oracle is a downstream task from the pre-trained TAPE transformer model \cite{tape2019}, more details about oracle could be found in \ref{sec:data_oracle}. For each round of the experiment, the model will query sequences from the black-box oracle, and the oracle will produce a fitness score depending on the sequence similar to the wet lab.

In particular, we use datasets of Adeno-associated virus 2 capsid protein (AAV) \cite{Bryant2021} which aimed for viral viability with search space $20^{28}$; TEM-1 $\beta$-Lactamase (TEM) \cite{Gonzalez2019-fr} which aimed for thermodynamic stability with search space $20^{286}$; and, Anti-SARS-CoV-2 antibody (AAYL49) \cite{Engelhart2022} which aimed for binding affinity with search space of $20^{118}$. A full description can be found in \ref{sec:dataset_full}

\subsubsection{Experiment Setup}

In the experiment, we run each algorithm for 10 rounds with 100 query sequences per round for a fair comparison with our baselines. The algorithm cannot get any information about the oracle except the channel of queried sequences. Each test is run for 50 repeats using 50 different random seeds and measured the average performance across all seeds.

We use PEX \citep{ren2022proximal}, Adalead \citep{sinai2020adalead}, Bayesian Optimization (BO) implemented by \citet{sinai2020adalead}, DbAS \citep{brookes2018design} implemented by \citet{sinai2020adalead} and LaMBO\citep{pmlr-v162-stanton22a} as our main baselines. We also use NeuralUCB-DE and NeuralTS-DE as the other two baselines which can be viewed as two downgraded versions of our algorithm without tree search. In detail, NeuralUCB-DE and NeuralTS-DE use uniform mutation with a mutation probability and recombination to generate new sequences. 
We also tested CbAS \citep{brookes2019conditioning} but choose not report it, because it performs almost the same as DbAS, consistent with results \citet{ren2022proximal}. The performance of the LaMBO on the TEM dataset is missing because the LaMBO algorithm has a memory limit and run-time efficiency issue on the TEM dataset due to the long sequence length. We do not use DyNAPPO \citep{Angermueller2020Model-based} as our baseline due to tensorflow incompatibility.

For tree search-based algorithms, we use a neural network and follow \cite{pan2022neuralbandits}'s approach and use the second last dense layer's output of the neural network as $\phi(x)$ used by the UCB/TS formula. We reported tree search-based algorithms both uses UCB(TreeNeuralUCB) and TS(TreeNeuralTS). More details are included in Appendix \ref{sec:data_oracle}.

We run all the exploration algorithms on the same neural network structure for fair comparison. One hot encoding with CNN was used for the TEM and AAYL49 antibody datasets. We use TAPE \citep{tape2019} embedding with a simple 2-layer fully connected neural network for the AAV dataset to compensate for the limited sequence length.

\subsection{Results and Analysis} 

\textbf{Performances and Comparison.}
The experiment results are shown in Figure \ref{fig:comparison_baseline}.  Our tree search-based algorithm generally outperforms our baselines regarding the max fitness performance. The detailed ranking for all algorithms on each dataset can be seen in Appendix \ref{sec:ranking_of_algorithm}. On the datasets TEM and AAYL49, all algorithms' curves seem not to reach convergence in 10 rounds. This is because TEM and AAYL49 deal with much longer sequences, so convergence is slower no matter what algorithm is used. We also conducted additional tests on the landscapes build by \citet{Thomas2022.10.28.514293} reported in \ref{sec:data_oracle}, the results are consistent with our oracle and the tree search method outperforms other baselines.

For the performance of mutation count, tree search-based algorithms are not worse than PEX-based algorithms. On the AAV dataset, all algorithms can be controlled within 5 mutation counts. However, on datasets with larger sequence lengths, both NeuralUCB-DE and NeuralTS-DE have very large mutation counts.

We also run an experiment on the 3gfb oracles build by \citet{Thomas2022.10.28.514293} using the same setting. We report the results in appendix \ref{sec:add_exp}. The results is consistent with other oracles and the tree search-based algorithms outperform the baselines regarding the max fitness performance.

\noindent\textbf{Effect of Tree Search compared to DE}
In Figure \ref{fig:comparison_baseline}, tree search-based algorithms generally outperform two baselines NeuralTS-DE and NeuralUCB-DE. Meanwhile, tree search-based algorithms only need fewer mutation counts than those non-treesearch methods.

  \begin{figure}[h]
  \vspace{-2mm}
  \centering
   \includegraphics[height = 0.44\linewidth]{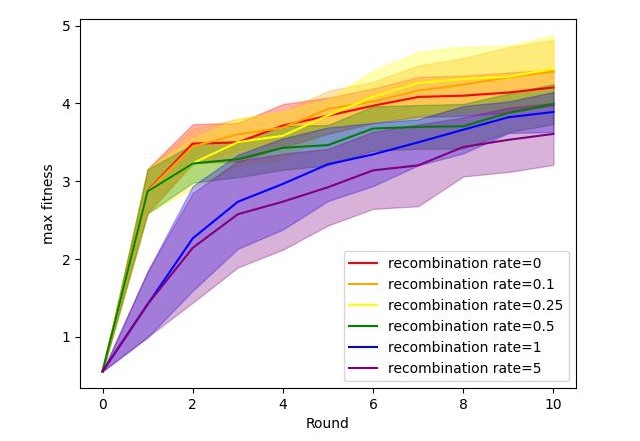} 
    \caption{Effects of recombination of TreeNeuralBandit tested with the AAV oracle. By varying the rate of recombination, we see that performance tops with a proper choice of positive recombination rate (yellow curve). }\label{fig:recombination}
    \vspace{-4mm}
  \end{figure}

\noindent\textbf{Effect of Recombination.}
We also compare different recombination rates. The line with a zero recombination rate means that the algorithm only includes mutation and doesn't consider the recombination process. In Figure \ref{fig:recombination}, we show that adding the recombination process to generating candidate children nodes is beneficial for improving performance, especially for small recombination rates. However, if the recombination rate becomes further larger, performance will gradually decrease and become even worse than the algorithm with a zero recombination rate.

  \begin{figure}
    \vspace{-2mm}
      \centering
      \includegraphics[height = 0.45\linewidth]{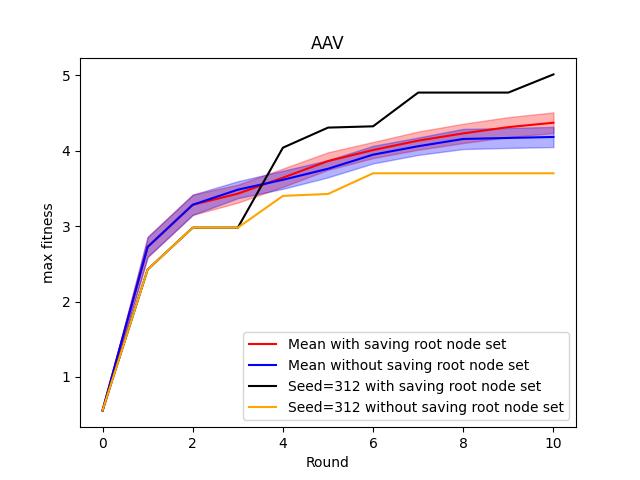} 
    \caption{Effect of keeping parent nodes in active set. }\label{fig:save_root_node_set}
      \label{fig:enter-label}
      \vspace{-4mm}
  \end{figure}

\noindent\textbf{Effect of keeping parent/root nodes in active set.}
We use a parameter $\rho$ to control updates of the active set, i.e., saving part of the root node set in each round. Figure \ref{fig:save_root_node_set} shows one seed's curve where the exploration got stuck when the active set is not updated to save enough parents/root. By saving previous nodes, it balances width-first and depth-first better and allows the algorithm to find higher-value nodes.

\noindent\textbf{Effect of diversity promotion.}
We tested adding an additional diversity bonus and observe a trade-off between group diversity and group fitness. Details are included in Appendix \ref{sec:diversity}.

\section{Limitations}

The functionalities of the protein sequences are complicated. So far, none of our oracles or oracles from Adalead and PEX-Mufacnet, or the pre-trained protein language models such as ESM or TAPE could accurately predict the whole landscape. 
We do not allow using our models in production without authorization due to potential negative social impacts.
More discussion on limitations is available in Appendix.
\section{Conclusion}\label{sec:conclusion}

This paper proposes a tree-based bandit learning method for enhancing the process of protein engineering. The methods expand the tree from the initial sequence (wild-type) with the guidance of bandit machine learning models using randomized tree search heuristics, machine learning models, pre-trained embeddings. We have shown the algorithm can discover a near-optimal design under simplified assumptions, with experimental validation.

\section*{Acknowledgments}
We extend our gratitude to MLAB Biosciences Inc. for their assistance in this work by providing abundant computing resources. We also extend our gratitude to Peiru Xu, who contributed to the implementation of several baseline algorithms during her internship at MLAB Biosciences.

\bibliography{main}

\begin{thebibliography}{56}
\providecommand{\natexlab}[1]{#1}

\bibitem[{Abbasi-Yadkori, P{\'a}l, and
  Szepesv{\'a}ri(2011)}]{abbasi2011improved}
Abbasi-Yadkori, Y.; P{\'a}l, D.; and Szepesv{\'a}ri, C. 2011.
\newblock Improved algorithms for linear stochastic bandits.
\newblock \emph{Advances in neural information processing systems}, 24:
  2312--2320.

\bibitem[{Angermueller et~al.(2020)Angermueller, Dohan, Belanger, Deshpande,
  Murphy, and Colwell}]{Angermueller2020Model-based}
Angermueller, C.; Dohan, D.; Belanger, D.; Deshpande, R.; Murphy, K.; and
  Colwell, L. 2020.
\newblock Model-based reinforcement learning for biological sequence design.
\newblock In \emph{International Conference on Learning Representations}.

\bibitem[{Arnold(1998)}]{arnold1998design}
Arnold, F.~H. 1998.
\newblock Design by directed evolution.
\newblock \emph{Accounts of chemical research}, 31(3): 125--131.

\bibitem[{Auer(2002)}]{Auer02}
Auer, P. 2002.
\newblock Using Confidence Bounds for Exploitation-Exploration Trade-offs.
\newblock \emph{Journal of Machine Learning Research}, 3: 397--422.

\bibitem[{Bedbrook et~al.(2017)Bedbrook, Yang, Rice, Gradinaru, and
  Arnold}]{bedbrook2017machine}
Bedbrook, C.~N.; Yang, K.~K.; Rice, A.~J.; Gradinaru, V.; and Arnold, F.~H.
  2017.
\newblock Machine learning to design integral membrane channelrhodopsins for
  efficient eukaryotic expression and plasma membrane localization.
\newblock \emph{PLoS computational biology}, 13(10): e1005786.

\bibitem[{Brookes, Park, and Listgarten(2019)}]{brookes2019conditioning}
Brookes, D.; Park, H.; and Listgarten, J. 2019.
\newblock Conditioning by adaptive sampling for robust design.
\newblock In \emph{International conference on machine learning}, 773--782.
  PMLR.

\bibitem[{Brookes and Listgarten(2018)}]{brookes2018design}
Brookes, D.~H.; and Listgarten, J. 2018.
\newblock Design by adaptive sampling.
\newblock \emph{arXiv preprint arXiv:1810.03714}.

\bibitem[{Bryant et~al.(2021)Bryant, Bashir, Sinai, Jain, Ogden, Riley, Church,
  Colwell, and Kelsic}]{Bryant2021}
Bryant, D.~H.; Bashir, A.; Sinai, S.; Jain, N.~K.; Ogden, P.~J.; Riley, P.~F.;
  Church, G.~M.; Colwell, L.~J.; and Kelsic, E.~D. 2021.
\newblock Deep diversification of an AAV capsid protein by machine learning.
\newblock \emph{Nature Biotechnology}, 39(6): 691--696.

\bibitem[{Chen and Arnold(1991)}]{chen1991enzyme}
Chen, K.; and Arnold, F.~H. 1991.
\newblock Enzyme engineering for nonaqueous solvents: random mutagenesis to
  enhance activity of subtilisin E in polar organic media.
\newblock \emph{Bio/Technology}, 9(11): 1073--1077.

\bibitem[{Chen and Arnold(1993)}]{chen1993tuning}
Chen, K.; and Arnold, F.~H. 1993.
\newblock Tuning the activity of an enzyme for unusual environments: sequential
  random mutagenesis of subtilisin E for catalysis in dimethylformamide.
\newblock \emph{Proceedings of the National Academy of Sciences}, 90(12):
  5618--5622.

\bibitem[{Chowdhury and Gopalan(2017)}]{chowdhury2017kernelized}
Chowdhury, S.~R.; and Gopalan, A. 2017.
\newblock On kernelized multi-armed bandits.
\newblock In \emph{International Conference on Machine Learning}, 844--853.
  PMLR.

\bibitem[{Doppa(2021)}]{doppa2021adaptive}
Doppa, J.~R. 2021.
\newblock Adaptive experimental design for optimizing combinatorial structures.
\newblock In \emph{Proceedings of the Thirtieth International Joint Conference
  on Artificial Intelligence (IJCAI)}, 4940--4945.

\bibitem[{Elnaggar et~al.(2020)Elnaggar, Heinzinger, Dallago, Rihawi, Wang,
  Jones, Gibbs, Feher, Angerer, Steinegger, Bhowmik, and
  Rost}]{Elnaggar2020ProtTrans}
Elnaggar, A.; Heinzinger, M.; Dallago, C.; Rihawi, G.; Wang, Y.; Jones, L.;
  Gibbs, T.; Feher, T.; Angerer, C.; Steinegger, M.; Bhowmik, D.; and Rost, B.
  2020.
\newblock ProtTrans: Towards Cracking the Language of Life's Code Through
  Self-Supervised Deep Learning and High Performance Computing.

\bibitem[{Engelhart et~al.(2022)Engelhart, Emerson, Shing, Lennartz, Guion,
  Kelley, Lin, Lopez, Younger, and Walsh}]{Engelhart2022}
Engelhart, E.; Emerson, R.; Shing, L.; Lennartz, C.; Guion, D.; Kelley, M.;
  Lin, C.; Lopez, R.; Younger, D.; and Walsh, M.~E. 2022.
\newblock A dataset comprised of binding interactions for 104,972 antibodies
  against a SARS-CoV-2 peptide.
\newblock \emph{Scientific Data}, 9(1): 653.

\bibitem[{Fannjiang and Listgarten(2020)}]{fannjiang2020autofocused}
Fannjiang, C.; and Listgarten, J. 2020.
\newblock Autofocused oracles for model-based design.
\newblock \emph{Advances in Neural Information Processing Systems}, 33:
  12945--12956.

\bibitem[{Fox et~al.(2003)Fox, Roy, Govindarajan, Minshull, Gustafsson, Jones,
  and Emig}]{fox2003optimizing}
Fox, R.; Roy, A.; Govindarajan, S.; Minshull, J.; Gustafsson, C.; Jones, J.~T.;
  and Emig, R. 2003.
\newblock Optimizing the search algorithm for protein engineering by directed
  evolution.
\newblock \emph{Protein engineering}, 16(8): 589--597.

\bibitem[{Fox et~al.(2007)Fox, Davis, Mundorff, Newman, Gavrilovic, Ma, Chung,
  Ching, Tam, Muley et~al.}]{fox2007improving}
Fox, R.~J.; Davis, S.~C.; Mundorff, E.~C.; Newman, L.~M.; Gavrilovic, V.; Ma,
  S.~K.; Chung, L.~M.; Ching, C.; Tam, S.; Muley, S.; et~al. 2007.
\newblock Improving catalytic function by ProSAR-driven enzyme evolution.
\newblock \emph{Nature biotechnology}, 25(3): 338--344.

\bibitem[{Freschlin, Fahlberg, and Romero(2022)}]{freschlin2022machine}
Freschlin, C.~R.; Fahlberg, S.~A.; and Romero, P.~A. 2022.
\newblock Machine learning to navigate fitness landscapes for protein
  engineering.
\newblock \emph{Current Opinion in Biotechnology}, 75: 102713.

\bibitem[{Gonzalez and Ostermeier(2019)}]{Gonzalez2019-fr}
Gonzalez, C.~E.; and Ostermeier, M. 2019.
\newblock Pervasive pairwise intragenic epistasis among sequential mutations in
  {TEM-1} $\beta$-lactamase.
\newblock \emph{J. Mol. Biol.}, 431(10): 1981--1992.

\bibitem[{Hao, Lattimore, and Wang(2020)}]{hao2020high}
Hao, B.; Lattimore, T.; and Wang, M. 2020.
\newblock High-dimensional sparse linear bandits.
\newblock \emph{Advances in Neural Information Processing Systems}, 33:
  10753--10763.

\bibitem[{Hibbert and Dalby(2005)}]{hibbert2005directed}
Hibbert, E.~G.; and Dalby, P.~A. 2005.
\newblock Directed evolution strategies for improved enzymatic performance.
\newblock \emph{Microbial Cell Factories}, 4(1): 1--6.

\bibitem[{Hsu et~al.(2022)Hsu, Nisonoff, Fannjiang, and
  Listgarten}]{Hsu2022Learning}
Hsu, C.; Nisonoff, H.; Fannjiang, C.; and Listgarten, J. 2022.
\newblock Learning protein fitness models from evolutionary and assay-labeled
  data.
\newblock \emph{Nature Biotechnology}, 40(7): 1114--1122.

\bibitem[{Jacot, Gabriel, and Hongler(2018)}]{jacot2018neural}
Jacot, A.; Gabriel, F.; and Hongler, C. 2018.
\newblock Neural tangent kernel: Convergence and generalization in neural
  networks.
\newblock \emph{Advances in neural information processing systems}, 31.

\bibitem[{Kuchner and Arnold(1997)}]{kuchner1997directed}
Kuchner, O.; and Arnold, F.~H. 1997.
\newblock Directed evolution of enzyme catalysts.
\newblock \emph{Trends in biotechnology}, 15(12): 523--530.

\bibitem[{Lansdell, Triantafillou, and Kording(2019)}]{lansdell2019rarely}
Lansdell, B.; Triantafillou, S.; and Kording, K. 2019.
\newblock Rarely-switching linear bandits: optimization of causal effects for
  the real world.
\newblock \emph{arXiv preprint arXiv:1905.13121}.

\bibitem[{Li et~al.(2022)Li, Wang, Wang, and Wang}]{li2022communication}
Li, C.; Wang, H.; Wang, M.; and Wang, H. 2022.
\newblock Communication efficient distributed learning for kernelized
  contextual bandits.
\newblock \emph{Advances in Neural Information Processing Systems}, 35:
  19773--19785.

\bibitem[{Li et~al.(2010)Li, Chu, Langford, and Schapire}]{li2010contextual}
Li, L.; Chu, W.; Langford, J.; and Schapire, R.~E. 2010.
\newblock A contextual-bandit approach to personalized news article
  recommendation.
\newblock In \emph{Proceedings of the 19th international conference on World
  wide web}, 661--670. ACM.

\bibitem[{Mockus(1989)}]{Mockus1989BayesianAT}
Mockus, J. 1989.
\newblock Bayesian Approach to Global Optimization: Theory and Applications.

\bibitem[{Nijkamp et~al.(2022)Nijkamp, Ruffolo, Weinstein, Naik, and
  Madani}]{nijkamp2022progen2}
Nijkamp, E.; Ruffolo, J.; Weinstein, E.~N.; Naik, N.; and Madani, A. 2022.
\newblock ProGen2: Exploring the Boundaries of Protein Language Models.

\bibitem[{Packer and Liu(2015)}]{packer2015methods}
Packer, M.~S.; and Liu, D.~R. 2015.
\newblock Methods for the directed evolution of proteins.
\newblock \emph{Nature Reviews Genetics}, 16(7): 379--394.

\bibitem[{Rao et~al.(2019)Rao, Bhattacharya, Thomas, Duan, Chen, Canny, Abbeel,
  and Song}]{tape2019}
Rao, R.; Bhattacharya, N.; Thomas, N.; Duan, Y.; Chen, X.; Canny, J.; Abbeel,
  P.; and Song, Y.~S. 2019.
\newblock Evaluating Protein Transfer Learning with TAPE.
\newblock In \emph{Advances in Neural Information Processing Systems}.

\bibitem[{Ren et~al.(2022)Ren, Li, Ding, Zhou, Ma, and Peng}]{ren2022proximal}
Ren, Z.; Li, J.; Ding, F.; Zhou, Y.; Ma, J.; and Peng, J. 2022.
\newblock Proximal exploration for model-guided protein sequence design.
\newblock In \emph{International Conference on Machine Learning}, 18520--18536.
  PMLR.

\bibitem[{Rives et~al.(2019)Rives, Meier, Sercu, Goyal, Lin, Liu, Guo, Ott,
  Zitnick, Ma, and Fergus}]{rives2019biological}
Rives, A.; Meier, J.; Sercu, T.; Goyal, S.; Lin, Z.; Liu, J.; Guo, D.; Ott, M.;
  Zitnick, C.~L.; Ma, J.; and Fergus, R. 2019.
\newblock Biological Structure and Function Emerge from Scaling Unsupervised
  Learning to 250 Million Protein Sequences.
\newblock \emph{PNAS}.

\bibitem[{Rives et~al.(2021)Rives, Meier, Sercu, Goyal, Lin, Liu, Guo, Ott,
  Zitnick, Ma, and Fergus}]{rives2021biological}
Rives, A.; Meier, J.; Sercu, T.; Goyal, S.; Lin, Z.; Liu, J.; Guo, D.; Ott, M.;
  Zitnick, C.~L.; Ma, J.; and Fergus, R. 2021.
\newblock Biological structure and function emerge from scaling unsupervised
  learning to 250 million protein sequences.
\newblock \emph{Proceedings of the National Academy of Sciences}, 118(15):
  e2016239118.

\bibitem[{Romero, Krause, and Arnold(2013)}]{romero2013navigating}
Romero, P.~A.; Krause, A.; and Arnold, F.~H. 2013.
\newblock Navigating the protein fitness landscape with Gaussian processes.
\newblock \emph{Proceedings of the National Academy of Sciences}, 110(3):
  E193--E201.

\bibitem[{Russo and Van~Roy(2014)}]{russo2014learning}
Russo, D.; and Van~Roy, B. 2014.
\newblock Learning to optimize via posterior sampling.
\newblock \emph{Mathematics of Operations Research}, 39(4): 1221--1243.

\bibitem[{Seeger(2004)}]{seeger2004gaussian}
Seeger, M. 2004.
\newblock Gaussian processes for machine learning.
\newblock \emph{International journal of neural systems}, 14(02): 69--106.

\bibitem[{Shin et~al.(2021)Shin, Riesselman, Kollasch, McMahon, Simon, Sander,
  Manglik, Kruse, and Marks}]{shin2021protein}
Shin, J.-E.; Riesselman, A.~J.; Kollasch, A.~W.; McMahon, C.; Simon, E.;
  Sander, C.; Manglik, A.; Kruse, A.~C.; and Marks, D.~S. 2021.
\newblock Protein design and variant prediction using autoregressive generative
  models.
\newblock \emph{Nature communications}, 12(1): 1--11.

\bibitem[{Shuai, Ruffolo, and Gray(2022)}]{Shuai2021.12.13.472419}
Shuai, R.~W.; Ruffolo, J.~A.; and Gray, J.~J. 2022.
\newblock Generative language modeling for antibody design.
\newblock \emph{bioRxiv}.

\bibitem[{Sinai et~al.(2020)Sinai, Wang, Whatley, Slocum, Locane, and
  Kelsic}]{sinai2020adalead}
Sinai, S.; Wang, R.; Whatley, A.; Slocum, S.; Locane, E.; and Kelsic, E. 2020.
\newblock AdaLead: A simple and robust adaptive greedy search algorithm for
  sequence design.
\newblock \emph{arXiv preprint}.

\bibitem[{Smith and Petrenko(1997)}]{smith1997phage}
Smith, G.~P.; and Petrenko, V.~A. 1997.
\newblock Phage display.
\newblock \emph{Chemical reviews}, 97(2): 391--410.

\bibitem[{Srinivas et~al.(2009)Srinivas, Krause, Kakade, and
  Seeger}]{srinivas2009gaussian}
Srinivas, N.; Krause, A.; Kakade, S.~M.; and Seeger, M. 2009.
\newblock Gaussian process optimization in the bandit setting: No regret and
  experimental design.
\newblock \emph{arXiv preprint arXiv:0912.3995}.

\bibitem[{Stanton et~al.(2022)Stanton, Maddox, Gruver, Maffettone, Delaney,
  Greenside, and Wilson}]{pmlr-v162-stanton22a}
Stanton, S.; Maddox, W.; Gruver, N.; Maffettone, P.; Delaney, E.; Greenside,
  P.; and Wilson, A.~G. 2022.
\newblock Accelerating {B}ayesian Optimization for Biological Sequence Design
  with Denoising Autoencoders.
\newblock In Chaudhuri, K.; Jegelka, S.; Song, L.; Szepesvari, C.; Niu, G.; and
  Sabato, S., eds., \emph{Proceedings of the 39th International Conference on
  Machine Learning}, volume 162 of \emph{Proceedings of Machine Learning
  Research}, 20459--20478. PMLR.

\bibitem[{Thomas et~al.(2022)Thomas, Agarwala, Belanger, Song, and
  Colwell}]{Thomas2022.10.28.514293}
Thomas, N.; Agarwala, A.; Belanger, D.; Song, Y.~S.; and Colwell, L.~J. 2022.
\newblock Tuned Fitness Landscapes for Benchmarking Model-Guided Protein
  Design.
\newblock \emph{bioRxiv}.

\bibitem[{Turner(2009)}]{turner2009directed}
Turner, N.~J. 2009.
\newblock Directed evolution drives the next generation of biocatalysts.
\newblock \emph{Nature chemical biology}, 5(8): 567--573.

\bibitem[{Wang et~al.(2022)Wang, Kim, Cong, and Wang}]{wang2022neural}
Wang, C.; Kim, J.; Cong, L.; and Wang, M. 2022.
\newblock Neural Bandits for Protein Sequence Optimization.
\newblock In \emph{2022 56th Annual Conference on Information Sciences and
  Systems (CISS)}, 188--193. IEEE.

\bibitem[{Williams and Rasmussen(2006)}]{williams2006gaussian}
Williams, C.~K.; and Rasmussen, C.~E. 2006.
\newblock \emph{Gaussian processes for machine learning}, volume~2.
\newblock MIT press Cambridge, MA.

\bibitem[{Winter et~al.(1994)Winter, Griffiths, Hawkins, and
  Hoogenboom}]{winter1994making}
Winter, G.; Griffiths, A.~D.; Hawkins, R.~E.; and Hoogenboom, H.~R. 1994.
\newblock Making antibodies by phage display technology.
\newblock \emph{Annual review of immunology}, 12(1): 433--455.

\bibitem[{Wu et~al.(2019)Wu, Kan, Lewis, Wittmann, and Arnold}]{wu2019machine}
Wu, Z.; Kan, S. B.~J.; Lewis, R.~D.; Wittmann, B.~J.; and Arnold, F.~H. 2019.
\newblock Machine learning-assisted directed protein evolution with
  combinatorial libraries.
\newblock \emph{Proceedings of the National Academy of Sciences}, 116(18):
  8852--8858.

\bibitem[{Xu et~al.(2020)Xu, Wen, Zhao, and Gu}]{pan2022neuralbandits}
Xu, P.; Wen, Z.; Zhao, H.; and Gu, Q. 2020.
\newblock Neural Contextual Bandits with Deep Representation and Shallow
  Exploration.
\newblock \emph{CoRR}, abs/2012.01780.

\bibitem[{Yang, Wu, and Arnold(2019)}]{yang2019machine}
Yang, K.~K.; Wu, Z.; and Arnold, F.~H. 2019.
\newblock Machine-learning-guided directed evolution for protein engineering.
\newblock \emph{Nature methods}, 16(8): 687--694.

\bibitem[{Yang and Wang(2019)}]{yang2019sample}
Yang, L.; and Wang, M. 2019.
\newblock Sample-optimal parametric q-learning using linearly additive
  features.
\newblock In \emph{International Conference on Machine Learning}, 6995--7004.
  PMLR.

\bibitem[{Yang and Wang(2020)}]{yang2020reinforcement}
Yang, L.; and Wang, M. 2020.
\newblock Reinforcement learning in feature space: Matrix bandit, kernels, and
  regret bound.
\newblock In \emph{International Conference on Machine Learning}, 10746--10756.
  PMLR.

\bibitem[{Yuan et~al.(2022)Yuan, Ni, Wang, Zhang, Cong, Szepesv{\'a}ri, and
  Wang}]{yuan2022bandit}
Yuan, H.; Ni, C.; Wang, H.; Zhang, X.; Cong, L.; Szepesv{\'a}ri, C.; and Wang,
  M. 2022.
\newblock Bandit Theory and Thompson Sampling-Guided Directed Evolution for
  Sequence Optimization.
\newblock \emph{arXiv preprint arXiv:2206.02092}.

\bibitem[{Zhang et~al.(2020)Zhang, Zhou, Li, and Gu}]{zhang2020neural}
Zhang, W.; Zhou, D.; Li, L.; and Gu, Q. 2020.
\newblock Neural thompson sampling.
\newblock \emph{arXiv preprint arXiv:2010.00827}.

\bibitem[{Zhou, Li, and Gu(2020)}]{zhou2020neural}
Zhou, D.; Li, L.; and Gu, Q. 2020.
\newblock Neural contextual bandits with ucb-based exploration.
\newblock In \emph{International Conference on Machine Learning}, 11492--11502.
  PMLR.

\end{thebibliography}

\appendix

\section{Related Work Details}\label{sec:related_work_details}
 Bayesian Optimization(BO)\cite{Mockus1989BayesianAT} is a classical method for optimizing protein sequences. \citet{sinai2020adalead} implements the code of an ensemble of models for Bayesian Optimization(BO) as one of its baselines. However, BO is computationally inefficient because it searches the entire sequence space in each round. LaMBO\citep{pmlr-v162-stanton22a} is also a BO-based algorithm that supports both single-objective and multi-objective, but it is again computationally inefficient on sequences with long lengths. Design by adaptive sampling (DbAS; \citet{brookes2018design}) uses a probabilistic framework as well as an adaptive sampling method to help explore the fitness landscape. Conditioning by adaptive sampling (CbAS; \citet{brookes2019conditioning}) is similar to DbAS but it also uses regularization to solve the pathologies for data far from training distribution. 
DyNA PPO\cite{Angermueller2020Model-based} uses proximal policy optimization for sequence design. However, DyNA PPO may search the space far from the wild type, which leads to a high mutation count. AdaLead\cite{sinai2020adalead} is an evolutionary algorithm with greedy selection: In each round, it selects the sequences near the best sequence derived in the last round on the fitness landscape to do recombination and mutation. Then it chooses the best sequence set according to the approximate model to query the ground-truth landscape model. 
PEX MufacNet\cite{ren2022proximal} is a local search method based on the principle of proximal optimization, with a MutFAC network tailored to mutated sequence data. It also follows a greedy strategy to select sequences to query and does not employ bandit-based exploration.

\section{Experiment Details}\label{sec:data_oracle}

\subsection{Bandit Exploration}
\label{subsec:ban}
We apply the bandit-based exploration mechanism to select sequences for fitness evaluation. 
In the UCB version, 
we use a parameter $\alpha$ to control the explore-exploit tradeoff and we calculate the UCB score using
$
    F_{\text{UCB}}(x) = \theta_{t - 1}^T \psi(x) + \alpha \sqrt{\psi^T(x) A_{t - 1}^{-1} \psi(x)},
$
where $\psi$ is some feature map, $\theta_t$ is the model parameter and the covariance matrix is defined as $M_t = \lambda\mathbf{I}+\sum_{i=1}^{t-1} \psi(x)\psi(x)^\top$. 
In the Thompson Sampling version, $F_t$ takes the form 
$
    F_{\text{TS}}(x) = \tilde \theta^T \psi(x),
$
where $\tilde \theta$ are parameters sampled from a posterior distribution $\mathcal{N}\left( \theta_{t - 1}, \alpha M_{t - 1}^{-1} \right)$.

We consider two predictive models in our tree search bandits, linear model, and neural networks.
For the linear model, we can simply take $\psi$ to be a pre-trained embedding $\phi=\psi$, and our experiment uses the TAPE \cite{tape2019} embedding, and $\theta_t$ is the closed form solution of ridge regression. For neural networks, we construct a neural net taking $\phi(x)$ as input and use the second last dense layer's output of the neural network as $\psi(x)$ for calculating bonus or posterior distribution, following the shallow exploration idea of \cite{pan2022neuralbandits}.

\subsection{Full Dataset Description} \label{sec:dataset_full}

\paragraph{Adeno-associated virus 2 capsid protein (AAV)} AAVs have been widely used in gene therapies. A particular interest is changing the sequence of the capsid protein to address the issues of tissue tropism and manufacturability. We collected dataset from \citet{Bryant2021}, which contains 283953 different length sequences engineered a region (positions 561-588) on the AAV2 VP1 protein for optimizing the viral viability. We limited the search space to the sequences which have the same length (28 amino acids) as the wildtype sequence and the search space is $20^{28}$.

\paragraph{TEM-1 $\beta$-Lactamase (TEM)} TEM-1 $\beta$-Lactamase is an E. coli protein that efficiently hydrolyzes penicillins and many cephalosporins, affording antibiotics resistance. It is used to understand the mutational effects and epistatic interactions between mutations in the fitness landscape. We collected the fitness data from \cite{Gonzalez2019-fr} which contains 16924 sequences. The sequence length is $286$ which results in $20^{286}$ search space.

\paragraph{Anti-SARS-CoV-2 antibody (AAYL49)} Recombinant antibody is a common therapeutic modality that is widely adopted in treating many diseases, including cancers, metabolic disorders, and viral infections. The Fv region of antibodies is the most important region for antigen binding, affording target specificity of the binding event. \citet{Engelhart2022} presented a dataset of quantitative binding scores for mutational variants of several antibodies against a SARS-CoV2 peptide. We use the AAYL49 heavy chain part of the dataset which includes 26454 sequences and the sequence length is 118 which results in $20^{118}$ search space.

\begin{table*}[th]
\caption{Dataset Summary}
\label{dataset_summary}
\vskip 0.15in
\begin{center}
\begin{small}
\begin{sc}
\begin{tabular}{c | c | c |c | c}
\toprule
Dataset & Aim & Size & WT Length & Search Space \\
\midrule
AAV \cite{Bryant2021}: & Viral Viability & 283953 & 28 & $20^{28}$ \\
TEM \cite{Gonzalez2019-fr}: & Thermodynamic Stability & 16926 & 286 & $20^{286}$ \\
AAYL49 \cite{Engelhart2022}: & Binding Affinity & 13922 & 118 & $20^{118}$ \\

\end{tabular}
\end{sc}
\end{small}
\end{center}
\vskip -0.1in
\end{table*}

\subsection{Data Preprocessing}

We preprocess each of the datasets to keep only the sequence and fitness information. The preprocessing of  each dataset is as follows:
\begin{itemize}
    \item \textbf{AAV: } We removed the sequence with type "stop" as the stop codon also remove the amino acids after the editable regions. The column "score" was used for fitness. The dataset after the preprocessing has 283953 sequences.
    \item \textbf{TEM: } We converted the mutated position to the full single and double mutation sequences using the wild-type and assign the corresponding single and double mutation fitness. The dataset after the preprocessing contains 16924 sequences
    \item \textbf{AAYL49: } We keep "MIT\_Target" and the "AAYL49" part of the dataset. We take the mean of the non-empty fitness score across different experimental repeats. We dropped the data with empty fitness scores for all three repeats. The preprocessed dataset contains 13922 sequences.
\end{itemize}

\subsection{Oracle Training}

For AAV and TEM datasets, we used the pretrained transformer model of TAPE \cite{tape2019} to generate a 768-dimension embedding as the input of the CNN model from \cite{sinai2020adalead} shown as in \ref{oracle_cnn}. We used a 80/20 train-test split for the preprocessed datasets with shuffle. The oracle was trained using the training set with learning rate $1 \times 10^{-4}$, batch size of $256$ for maximum of $3000$ rounds with an Adam optimizer. The training will be early stopped if the loss does not decrease for $100$ epochs.

\begin{table}[h]
\caption{CNN Architecture used by Oracle}
\label{oracle_cnn}
\vskip 0.15in
\begin{center}
\begin{small}
\begin{sc}
\begin{tabular}{c c}
\toprule
\textbf{Input}: & TAPE Embedding  \\
\midrule
1D - CNN & Number of Filters: 32, Kernal Size: 5 \\
& ReLU Activation \\
1D - CNN & Number of Filters: 32, Kernal Size: 5 \\
& ReLU Activation \\
1D - CNN & Number of Filters: 32, Kernal Size: 5 \\
& ReLU Activation \\
& Global Max Pooling \\
Dense & Output Size: 256 \\
& ReLu Activation \\
Dense & Output Size: 256 \\
& ReLu Activation \\
& Dropout: $p = 0.25$ \\
Dense & Output Size : 1
\end{tabular}
\end{sc}
\end{small}
\end{center}
\vskip -0.1in
\end{table}

For the AAYL dataset, the method used for AAV and TEM results in a low spearman correlation of $0.273$. Hence, we trained a downstream task from the pretrained TAPE transformer model from \citet{tape2019}. We also used the same used the same 80/20 train-test split and train the model with parameter suggested by \citet{tape2019} with batch size of 32, learning rate $1 \times 10^{-5}$ with linear warm up until loss does not decrease for 10 rounds. 

We tested the correlation between the oracle prediction and ground-truth fitness on the test set which shown in \ref{oracle_corr}

\begin{table}[h]
\caption{Oracle Correlations}
\label{oracle_corr}
\vskip 0.15in
\begin{center}
\begin{small}
\begin{sc}
\begin{tabular}{c | c | c |c}
\toprule
Dataset: & AAV & TEM & AAYL49 \\
\midrule
Spearman: & 0.838 & 0.713 & 0.502 \\
Pearson: & 0.876 & 0.681 & 0.584
\end{tabular}
\end{sc}
\end{small}
\end{center}
\vskip -0.1in
\end{table}

\subsection{Experiment Details}

We run all the exploration algorithms on the same neural network structure for fair comparison. Onehot encoding with CNN was used for the TEM and AAYL49 antibody datasets. We use TAPE \citep{tape2019} embedding with simple 2 layer fully connected neural network for the AAV dataset to compensate the limited sequence length.

All algorithms was run on NVIDIA A10G GPU for 50 different random seeds. With 4-5 algorithms running in parallel, the total time is around 6 hours. The neural network is run with learning rate of $10^{-4}$, batch size of $64$ for maximum of 3000 epochs with a early stopping when loss does not decreases for 50 rounds. The baseline algorithms was run using their recommended setting.

\subsection{Additional Experiments on Other Landscape} \label{sec:add_exp}

We also run experiment on the 3gfb oracle designed by \citet{Thomas2022.10.28.514293} for all algorithm using the same setting as \ref{sec:exp} on the same model of onehot encoding and CNN. The maximum fitness curve with the mutation count is shown in \ref{fig:3gfb}. For convience, the legend of the fitness plot was in descending order regarding max fitness performance. We could see the performance is consistent with the results on other oracles with TreeNeuralTS and TreeNeuralUCB outperform other exploration algorithms. The mutation counts of TreeNeuralTS and TreeNeuralUCB are comparable to other algorithms and are reletively low compare to the wild-type sequence length of 347. The mutation count and fitness of BO are both zero as it fail to find any sequence with higher fitness than wild-type so the maximum fitness sequence was the initial wild-type sequence for all rounds. 

\begin{figure*}[t]
    \centering
    \includegraphics[width = 0.49\textwidth, height = 0.4\textwidth]{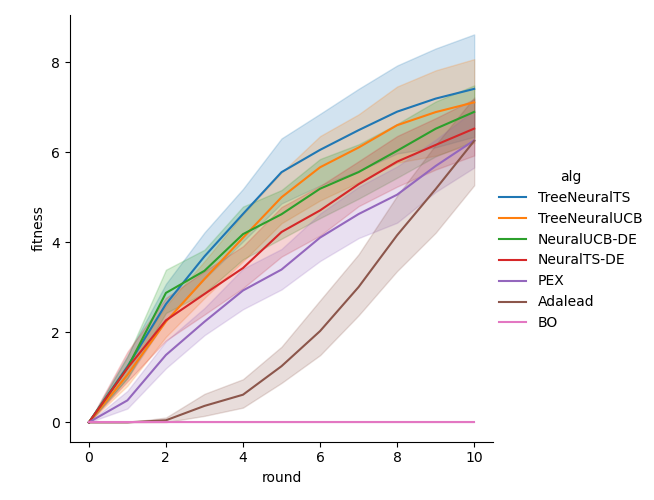}
    \includegraphics[width = 0.49\textwidth, height = 0.4\textwidth]{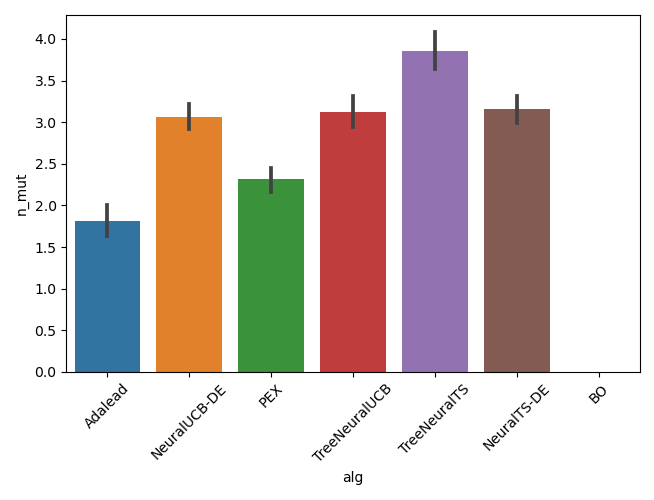}
    \caption{Experiment Result on 3gfb \citep{Thomas2022.10.28.514293} Oracle}
    \label{fig:3gfb}
\end{figure*}

\subsection{Simulation Results and Details}
We are also interested in testing our algorithms on synthetic data to bridge the gap between theory and. We use Gaussian Process to generate synthetic data and build a new oracle\citep{williams2006gaussian}. For consistency, our simulation is controlled within 10 rounds. In Figure \ref{fig:comparison_simulation}, we can find that tree-based algorithms can achieve higher fitness scores in 10 rounds and are closer to the maximum fitness than those algorithms which are not tree-based. Also, in the left panel of Figure \ref{fig:comparison_simulation}, tree-based algorithms' regret converges and maintains a low level while non-tree algorithms have the trend to increase in future rounds like LinUCB-DE or have higher regret.

\begin{figure*}[t]
\centering
   \includegraphics[width=0.49\textwidth]{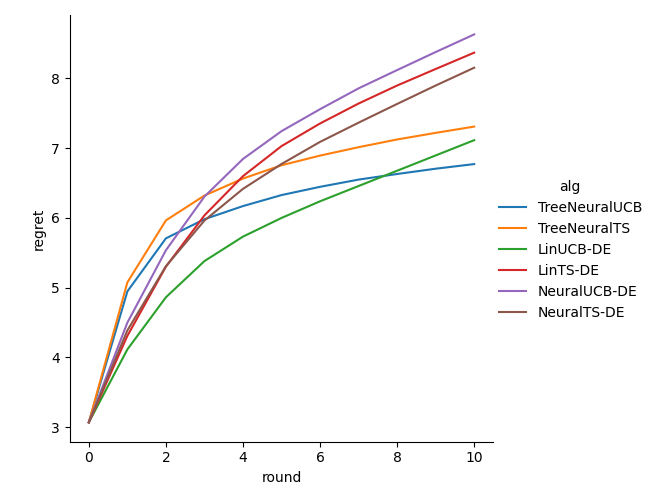}   
   \includegraphics[width=0.49\textwidth]{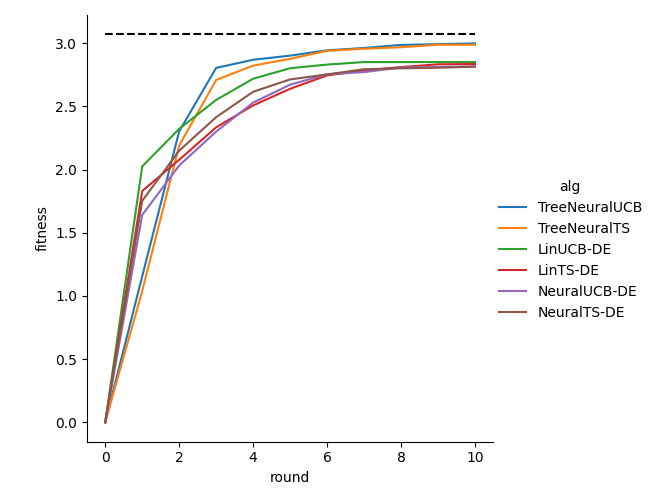}
      \caption{Regret and fitness curves of algorithms with comparison to baselines over synthetic data.}\label{fig:comparison_simulation}
\end{figure*}

\newpage

\subsection{Code Availability}
We make our code public anonymously to help confirm the details. Our baseline methods like BO, PEX, Adalead, CbAS and DbAS are implemented by their authors and we directly use their code. The anonymous link to our code is https://anonymous.4open.science/r/TreeSearchDE-public-0D26/README.md

\newpage

\section{Visualization of the Full Tree Search Process by Round}\label{sec:full_search_path}

\begin{figure*}[h]
 \centering
    \subfloat[$t = 1$]{\includegraphics[width=0.3\textwidth]{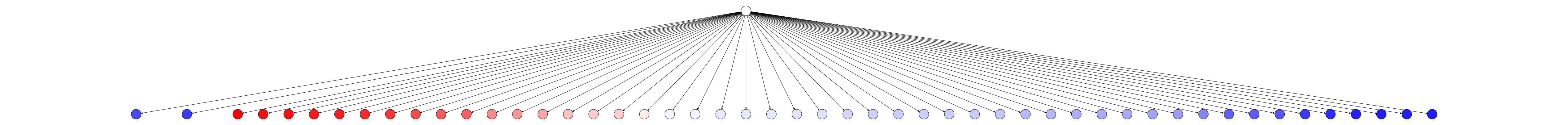}}
    \subfloat[$t = 2$]{\includegraphics[width=0.3\textwidth]{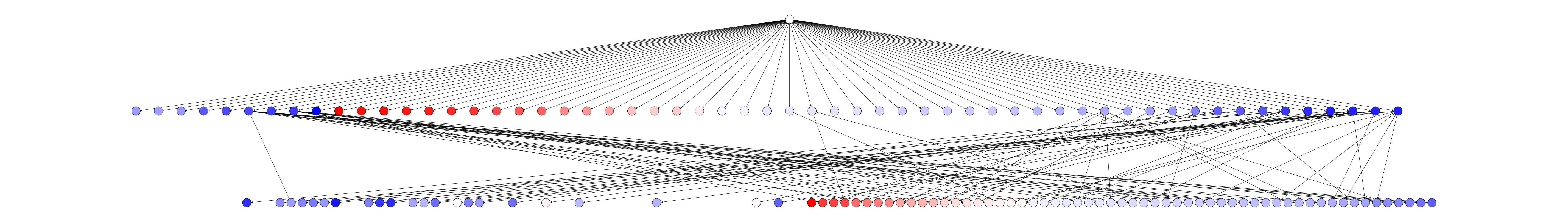}}
    \subfloat[$t = 3$]{\includegraphics[width=0.3\textwidth]{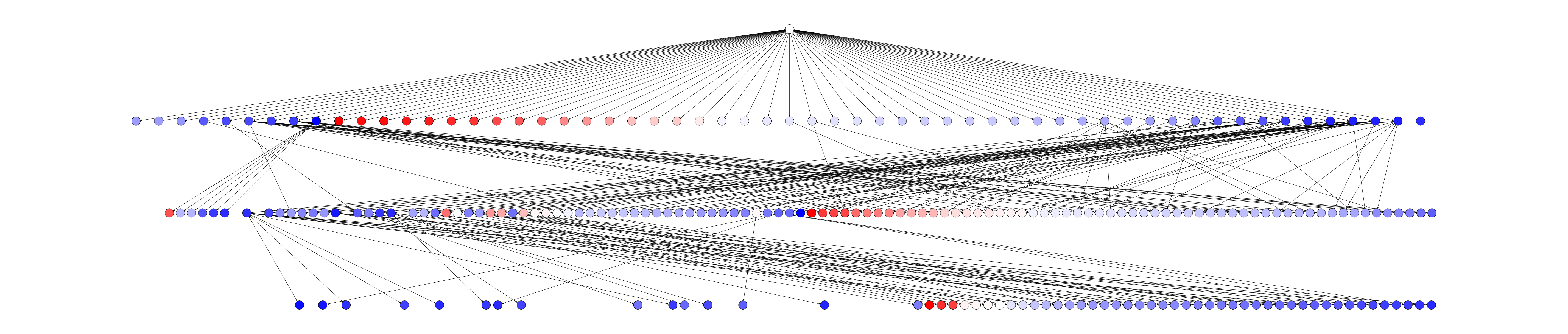}}
    
    \subfloat[$t = 4$]{\includegraphics[width=0.3\textwidth]{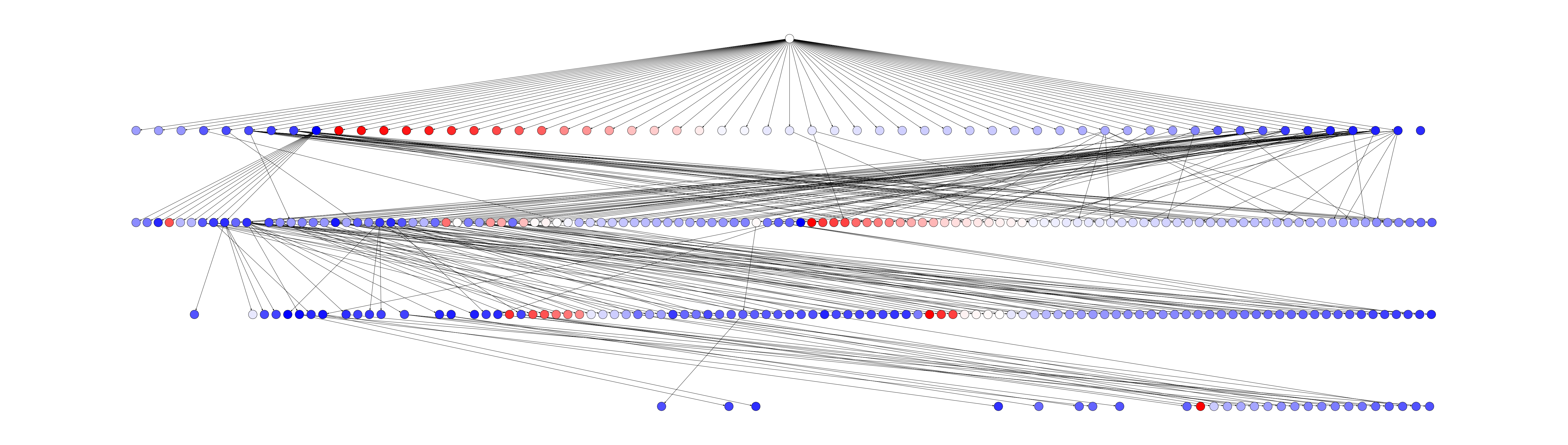}}
    \subfloat[$t = 5$]{\includegraphics[width=0.3\textwidth]{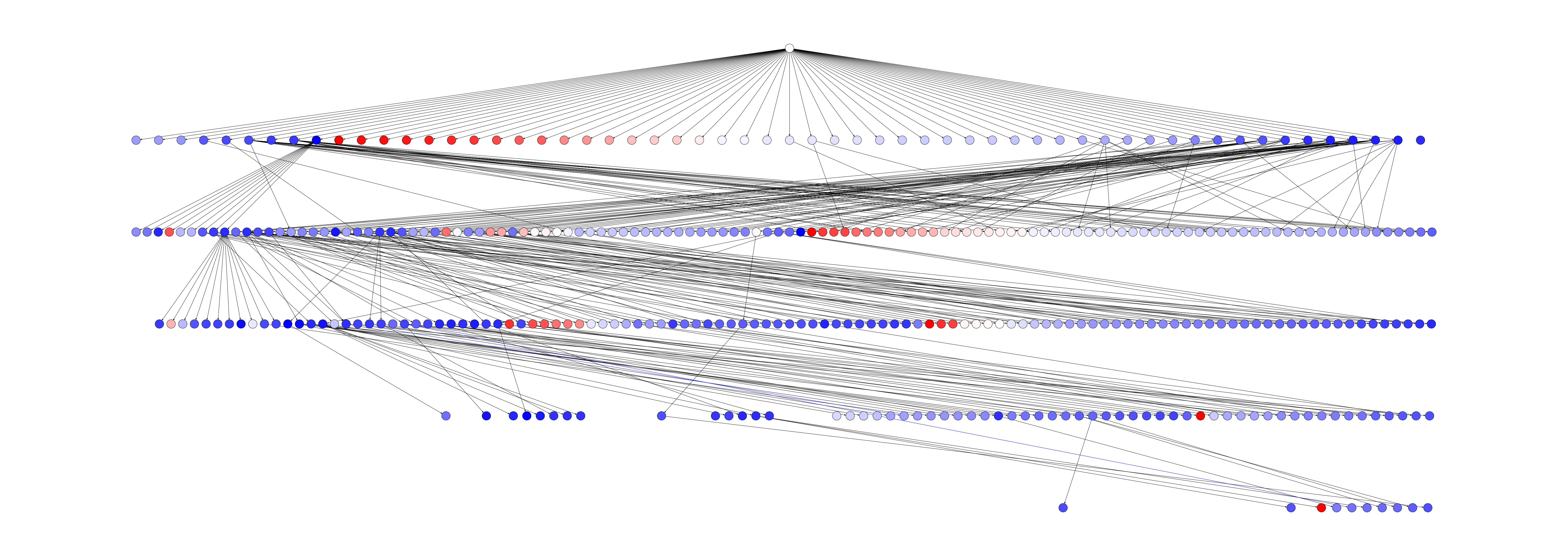}}
    \subfloat[$t = 6$]{\includegraphics[width=0.3\textwidth]{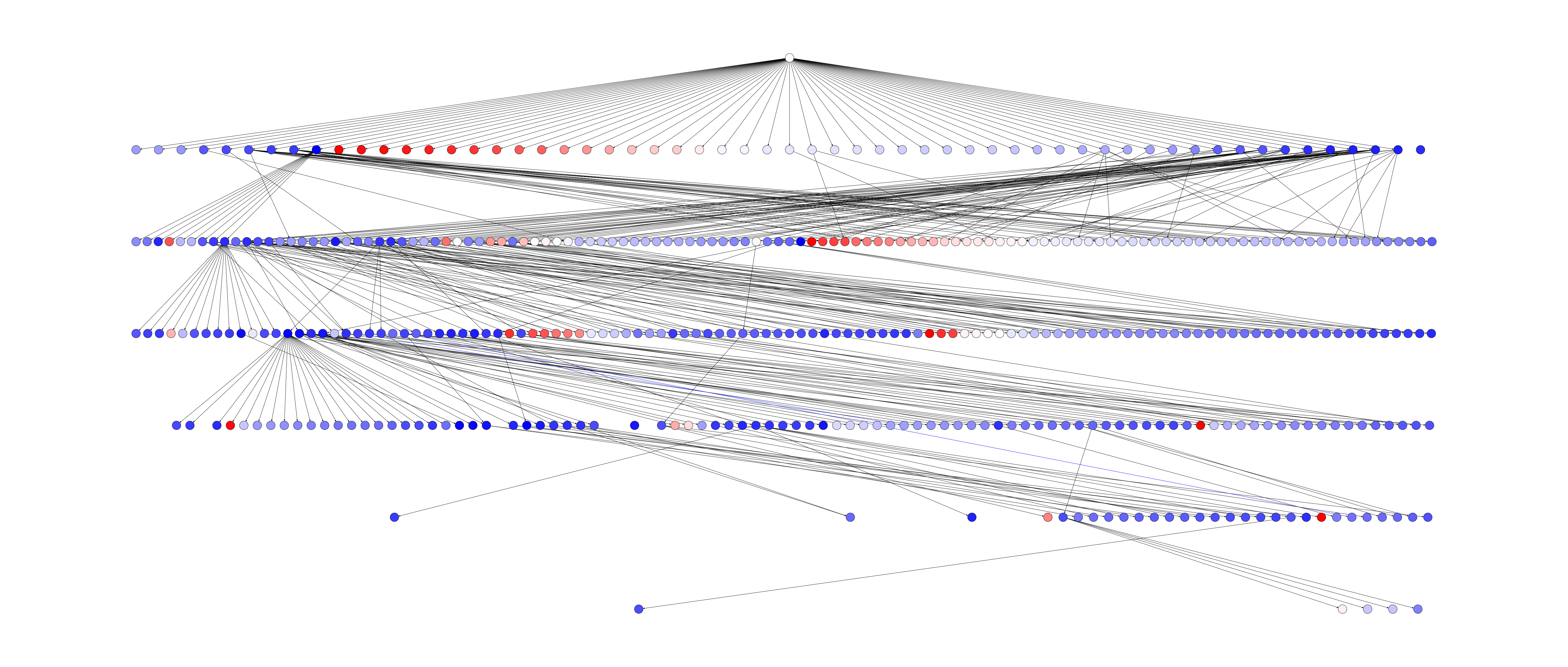}}
    
    \subfloat[$t = 7$]{\includegraphics[width=0.3\textwidth]{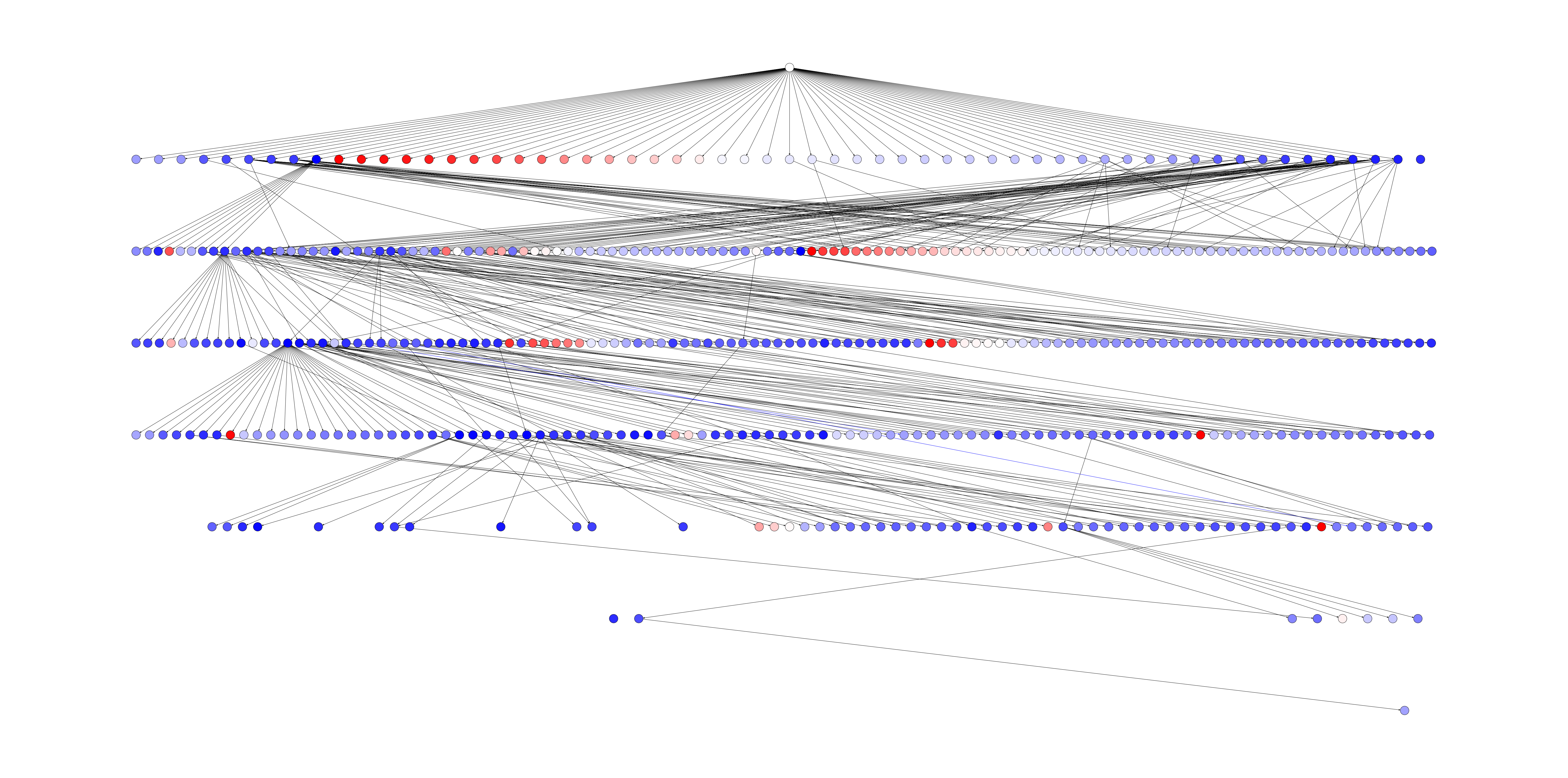}}
    \subfloat[$t = 8$]{\includegraphics[width=0.3\textwidth]{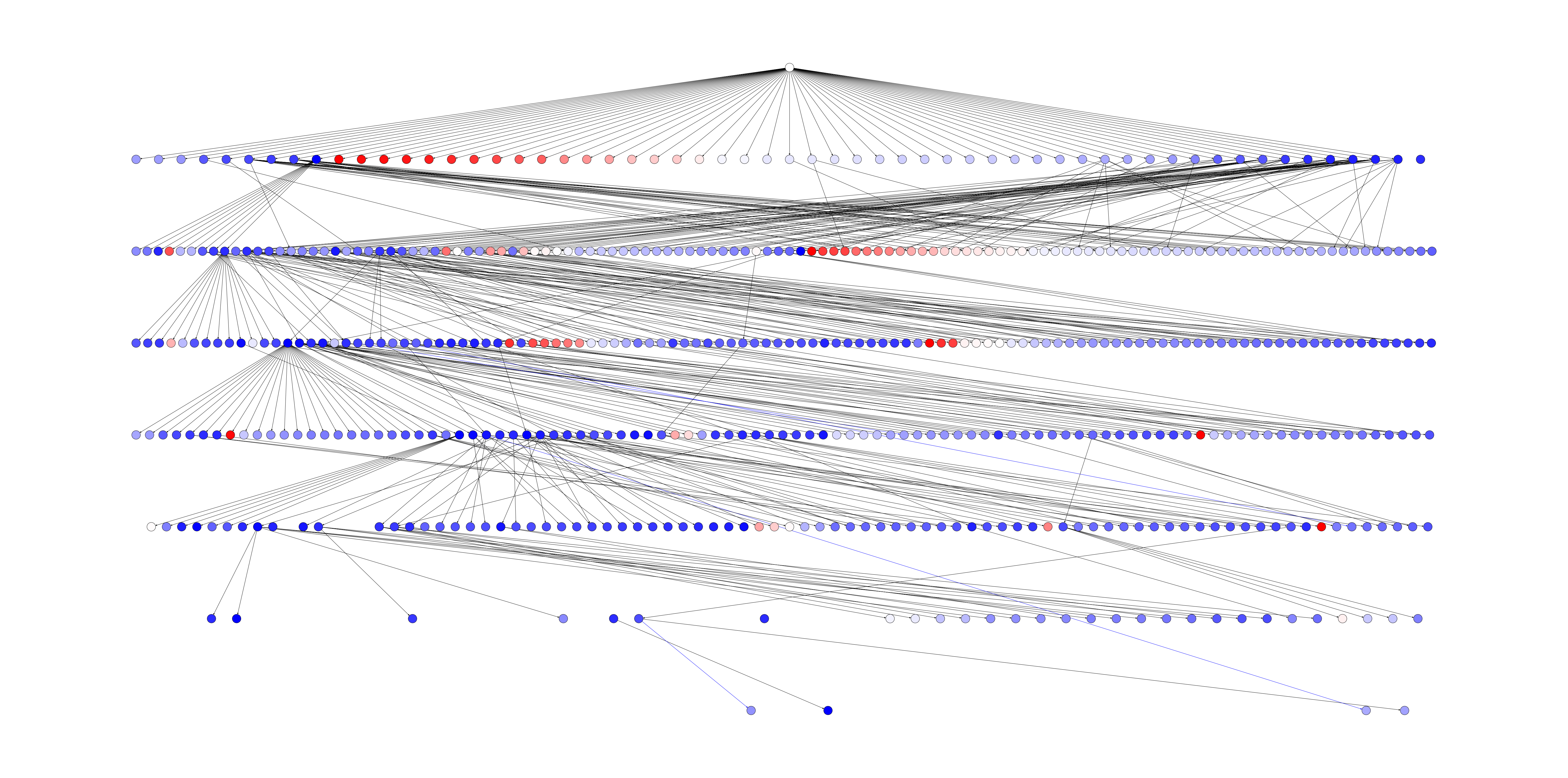}}
    \subfloat[$t = 9$]{\includegraphics[width=0.3\textwidth]{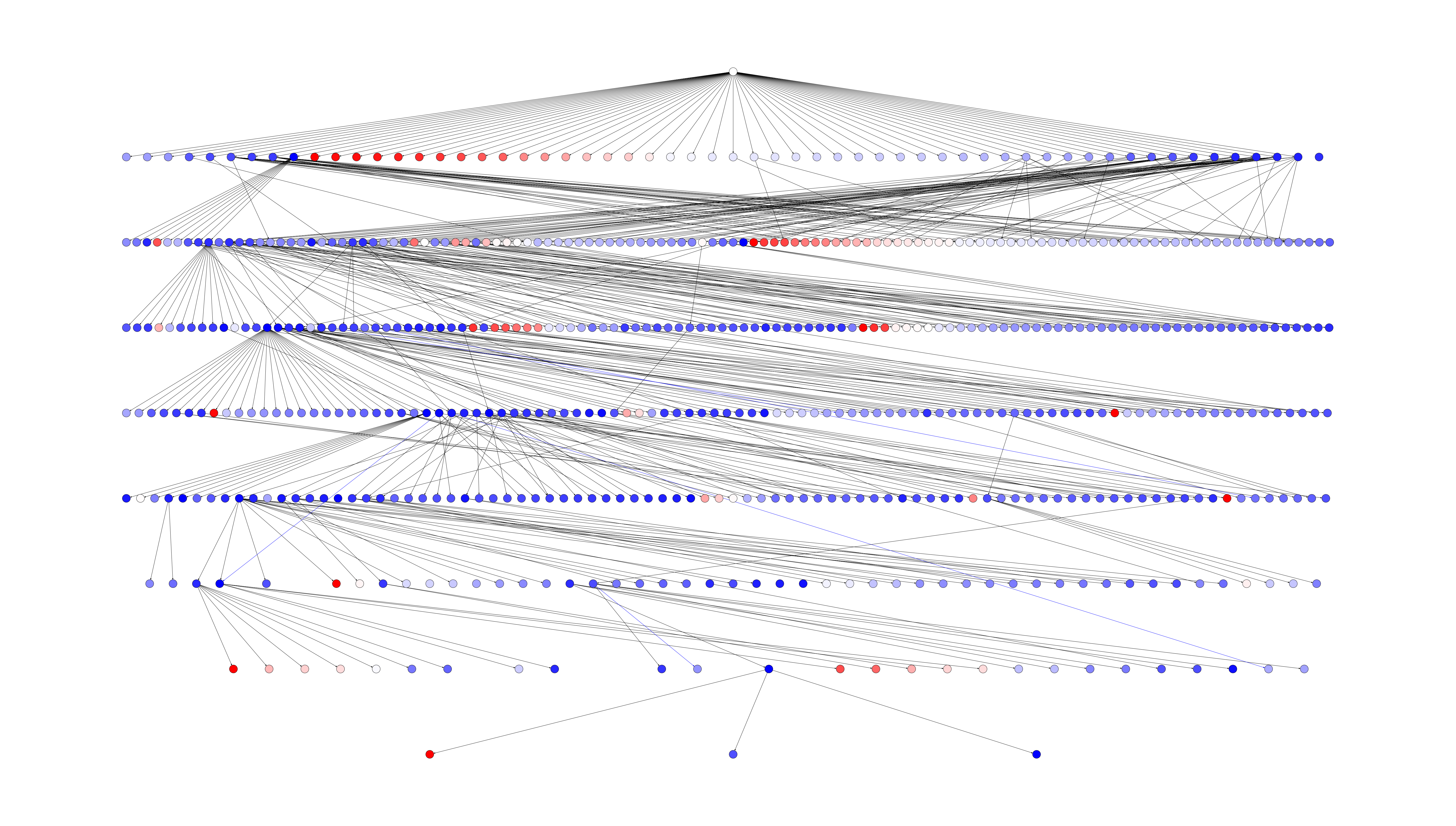}}
    
    \subfloat[$t = 10$]{\includegraphics[width=\textwidth]{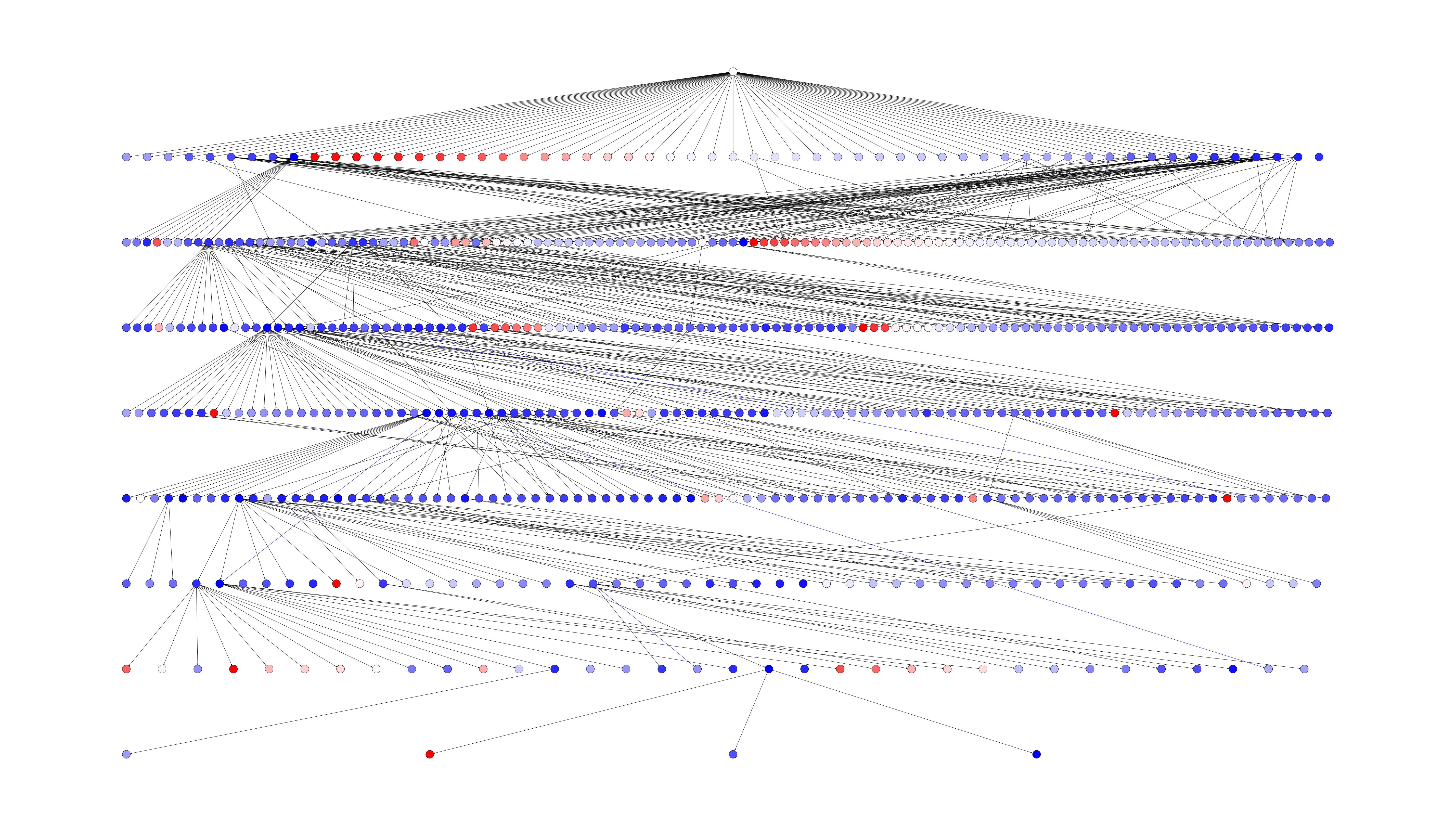}}
      \caption{Visualization of the full tree search process by round}\label{fig:full_tree}
\end{figure*}

\clearpage

\section{Diversity-Fitness Trade Off}\label{sec:diversity}

We aim to improve the diversity of our algorithm by adding a diversity constraint. The diversity score of the children node is defined as its distance to the center of nodes explored in the previous rounds in the embedding space, i.e.,
\begin{equation}
    S_{diversity} = \sqrt{\sum_{i=1}^n {(\phi(x)_i - {\overline{\phi(x^{explored})}}_i)^2}},
\end{equation}
where $\overline{\phi(x^{explored})}$ is the center of explored sequences in the embedding space and n is the dimension of the embedding space.
Then the score function for determining the candidate node set is modified as:
\begin{equation}
    S_{total} = S_{origin} + q * S_{diversity},
\end{equation}
where $q$ is the weight for $S_{diversity}$.

\subsection{Experiment}\label{sec:diversity_exp}

We tested the effect of different parameters of $q$'s on both diversity and fitness. In particular, we run the algorithm with $q \in \{0, 0.01, 0.03, 0.05, 0.08, 0.1, 0.3, 0.5, 0.8, 1, 2, 5, 10\}$ each for 10 random seeds. From \ref{fig:diversity}, we observed a trade-off between fitness and diversity, i.e., a higher $q$ will lead to a higher diversity but at the same time a lower average fitness.

\begin{figure*}[t]
    \centering
    \subfloat[Average Fitness]{\includegraphics[width = 0.5\textwidth]{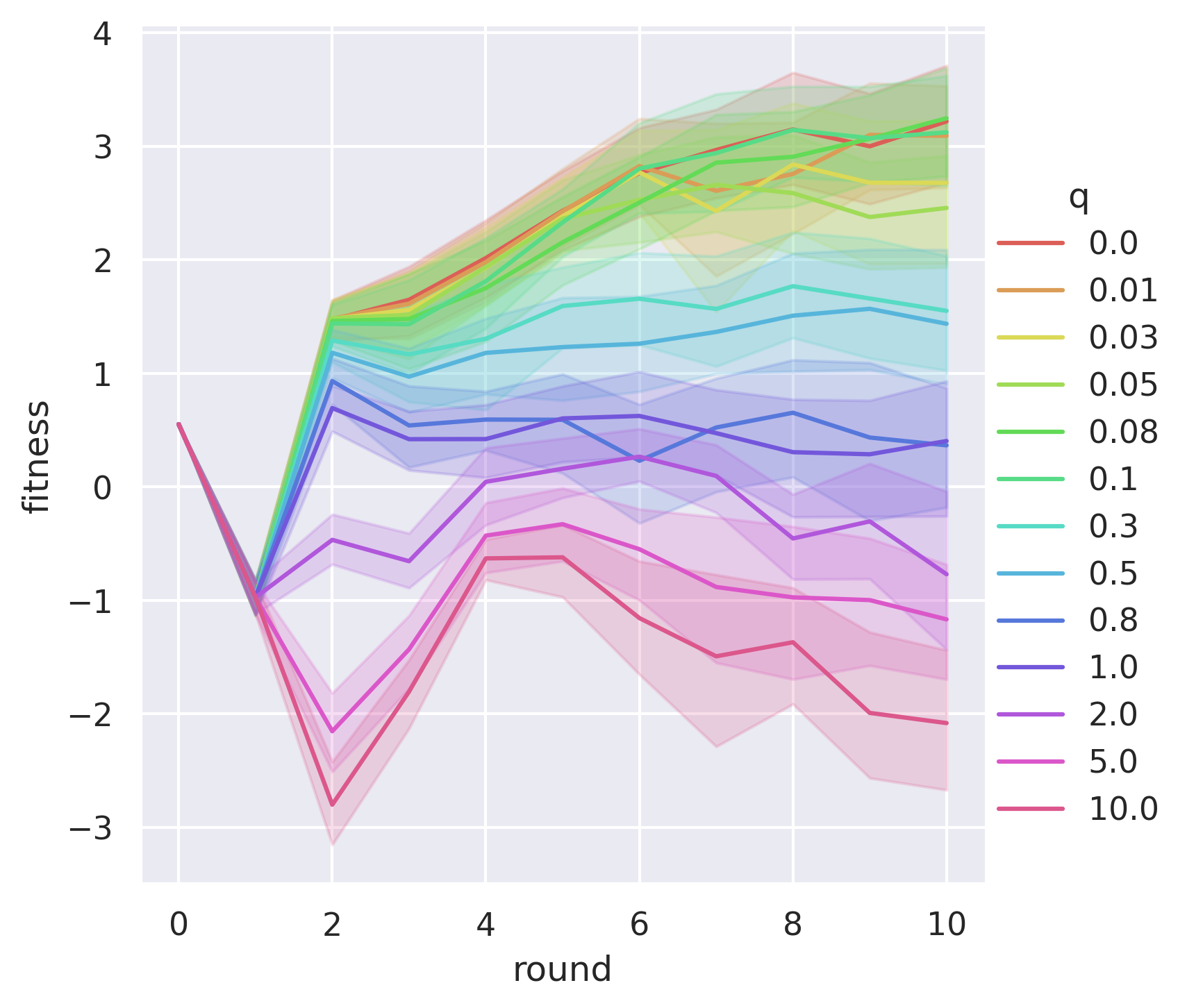}}
    \subfloat[Diversity Score]{\includegraphics[width = 0.5\textwidth]{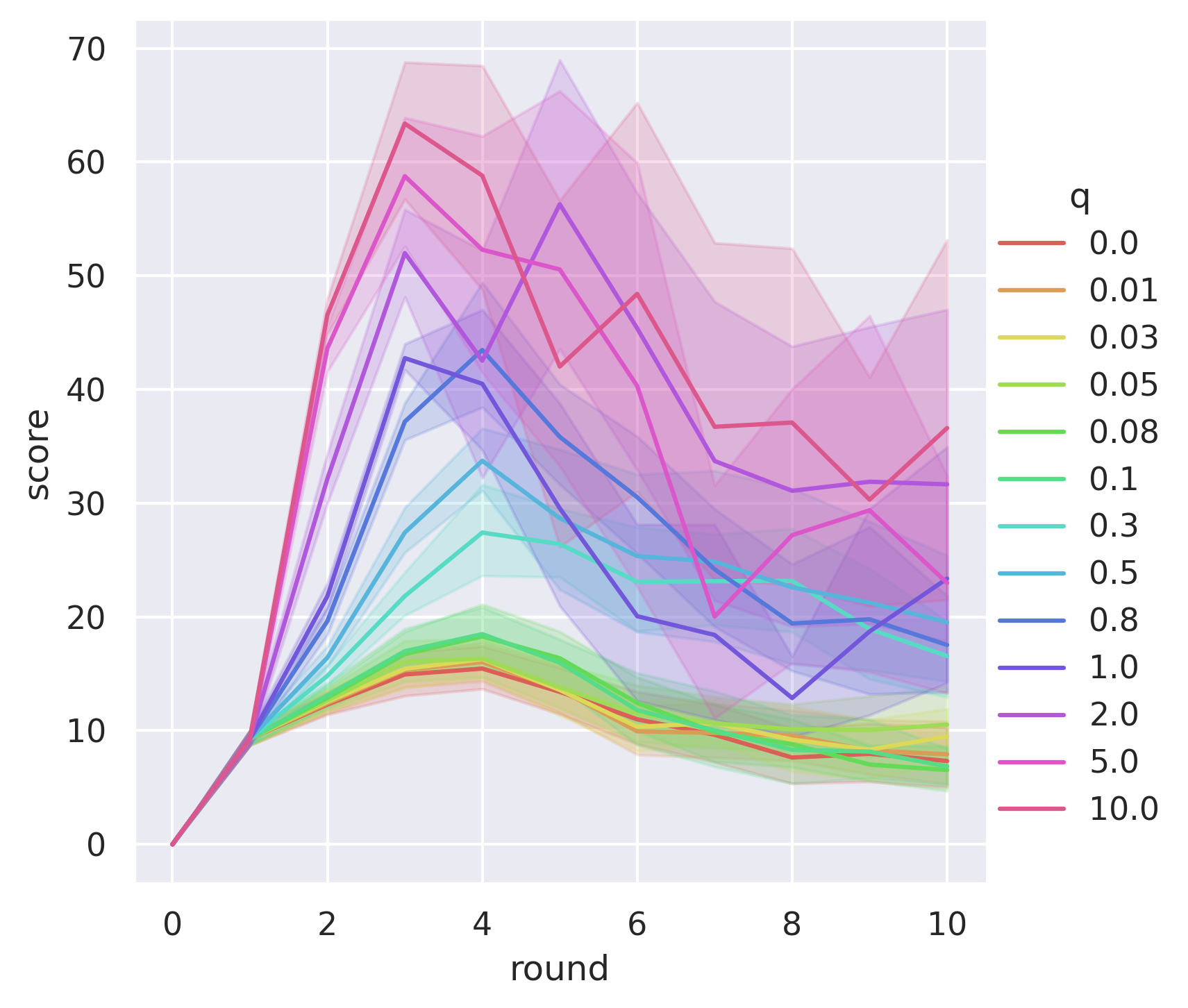}}
    \caption{Experiment Results for Different $q$'s}
    \label{fig:diversity}
\end{figure*}

\subsection{t-SNE}

We plot a t-SNE plot for all $q$'s for one particular random seed shown in Figure \ref{fig:tsne_diversity}. The t-SNE plot verifies the tread-off we discussed in \ref{sec:diversity_exp}. As the $q$ increases, the queried sequences are more spread out on the t-SNE plot which verifies the improvement in diversity. In contrast, when the $q$ is low, the queried sequences are concentrated on a smaller region of the t-SNE plot (local optimal) especially in the later round (shown in red color), and have a better fitness caused by high fitness sequences in this local optimal region.

\begin{figure*}[h]
    \centering
    \includegraphics[width = \textwidth]{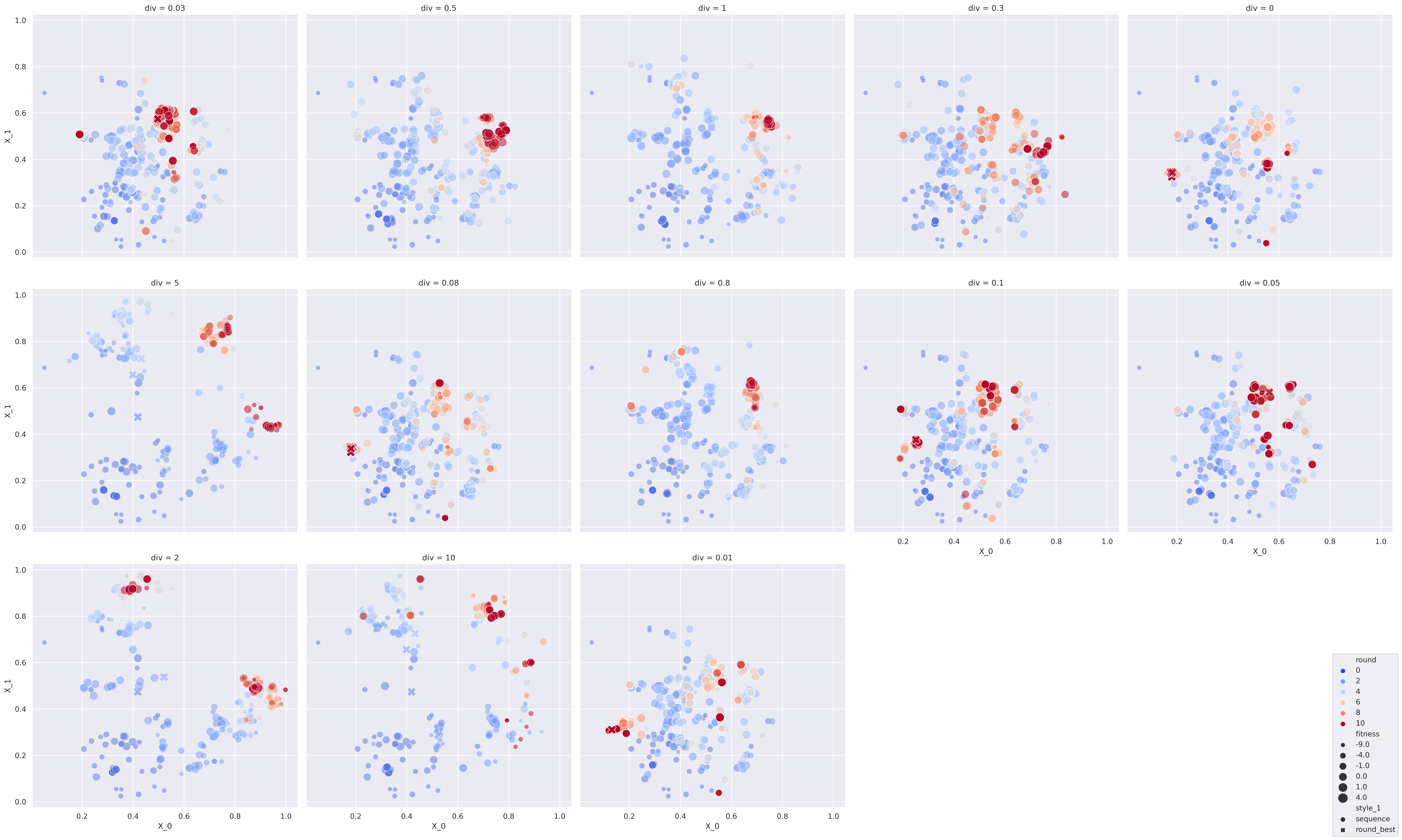}
    \caption{t-SNE Plot for Different $q$'s}
    \label{fig:tsne_diversity}
\end{figure*}

\section{Ranking of Algorithms}\label{sec:ranking_of_algorithm}

In Figure \ref{fig:comparison_baseline_with_rank}, the legend ranks the algorithm in descending order regarding max fitness performance for convenience. In Table \ref{table:performance_detail}, we show the detailed values for all algorithms.

\begin{table}[h]
\caption{Average Max Fitness of Algorithms}
\label{table:performance_detail}
\vskip 0.15in
\begin{center}
\begin{small}
\begin{sc}
\begin{tabular}{c | c | c |c}
\toprule
Dataset: & AAV & AAYL49 & TEM \\
\midrule
TreeNeuralUCB: & 4.3723 & -0.4976 & 2.394493 \\
TreeNeuralTS:  & 4.3056 & -0.4132 & 2.356243 \\
NeuralTS-DE:   & 3.7363 & -0.4773 & 2.108297 \\
NeuralUCB-DE:  & 3.7697 & -0.4421 & 1.978440 \\
PEX:           & 3.9286 & -0.4482 & 2.283808 \\
Adalead:       & 1.8835 & -0.9322 & 2.225310 \\
BO:            & 0.9299 & -1.0038 & 1.215830 \\
LaMBO:         & 3.7085 & -0.5813 & N/A \\
DbAS:          & 2.8481 & -0.7580 & 1.035354
\end{tabular}
\end{sc}
\end{small}
\end{center}
\vskip -0.1in
\end{table}

\begin{figure*}[h]
\centering
   \includegraphics[width=0.48\textwidth]{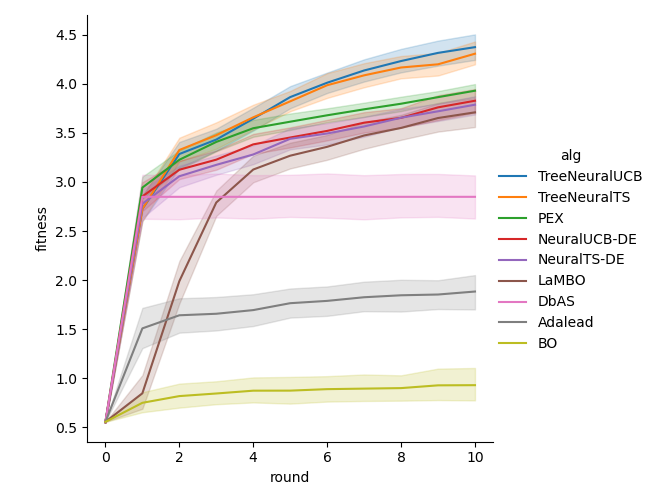}  
   \includegraphics[width=0.48\textwidth]{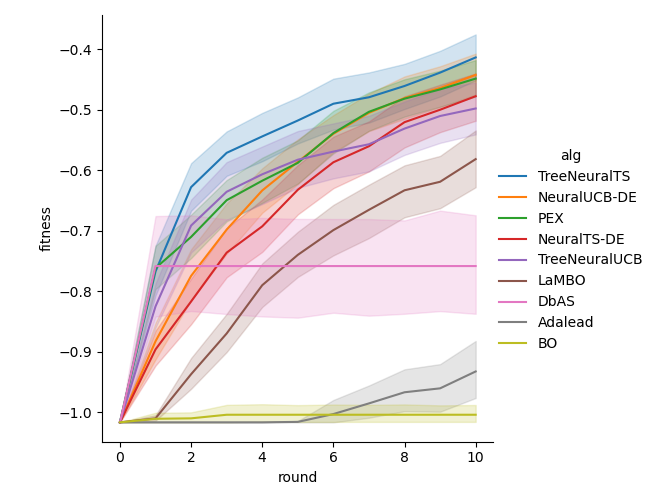}
    \includegraphics[width=0.48\textwidth]{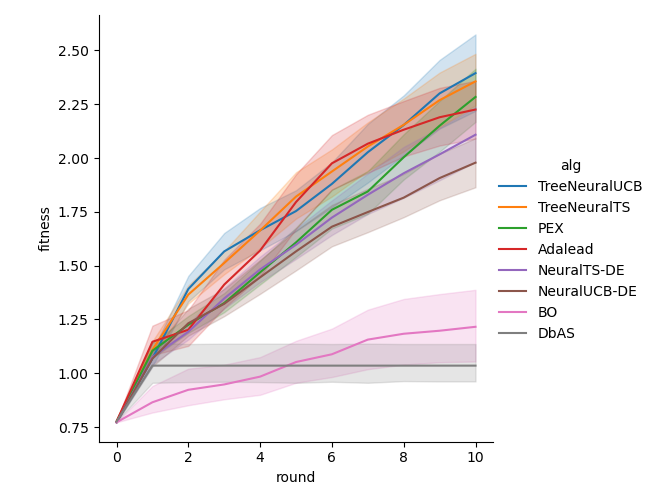}
      \caption{TreeBandit Compared to baselines}\label{fig:comparison_baseline_with_rank}
\end{figure*}

\section{Discussion on Limitations}

Our theoretical result serves its purpose well to provide an understanding of using local mutation to optimize protein utility. 
Our bound is a Bayesian regret bound. Also, it is an open question to derive a frequentist bound for, which is of interest in practice when ground truth f is agnostic. Current Bayesian regret does not directly apply to a frequentist one.

The functionalities of the protein sequences are complicated. So far, none of our oracles or oracles from Adalead and PEX-Mufacnet, or the pre-trained protein language models such as ESM or TAPE could accurately predict the whole landscape. However, we believe the significance of our experiment is not to build an oracle that could perfectly represent the landscape but instead to demonstrate a new exploration algorithm that could be applied to all protein landscapes. As from the experiments that our algorithm performs well on all oracles we built as well as on the new dataset from antibodies, we should expect our algorithm will perform well on the actual landscape.

Also, we are actively building our wet lab to apply our algorithm into practice, but it is acknowledged that wet-lab experiment takes a very long time for validation.

We do not allow using our models in production without authorization due to potential negative social impacts.
\section{Theory Details}
\subsubsection{Preliminaries of GP} 
Conditioned the set of screen data $\{(\phi(x_{i}), \tilde F(x_{i}))\}_{i=1}^{t-1}$ collected up to step $t$, the posterior distribution of $f^{\star}$ corresponding to the prior $\mathcal{GP}\rbra{0, k(\cdot, \cdot)}$ is still a Gaussian process $\mathcal{G P}(m_t, k_t)$.
\begin{definition}[Posterior $\mathcal{G P}(m_t, k_t)$]
    $m_t: \mathbb{R}^d \rightarrow \mathbb{R}$ and $k_t: \mathbb{R}^d \times \mathbb{R}^d \rightarrow \mathbb{R}$ are defined on the embedding space and have the form:
    \begin{align}
    \label{equ_post_mean}
    m_t(\phi(x)) &= \mathbf{k_{t-1}}(\phi(x))^{\top} \left(K_{t-1}+ \lambda I \right)^{-1} \tilde F_{t-1},\\
    \nonumber k_t \left(\phi(x), \phi(x^{\prime})\right) &= k\left(\phi(x), \phi(x^{\prime})\right) \\
    \label{equ_post_cov}& -\mathbf{k_{t-1}}(\phi(x))^{\top} \left(K_{t-1}+ \lambda I \right)^{-1}\mathbf{k_{t-1}}(\phi(x^{\prime})),
    \end{align}
    where $\mathbf{k_{t}}(\phi(x))=\left(k\left( \phi(x_1), \phi(x) \right), \cdots, k\left( \phi(x_{t}), \phi(x) \right)\right)^{\top}$, $K_t$ is the kernel matrix $\sbra{k(\phi(x) , \phi(x^{\prime}) )}_{x,x^{\prime}\in A_t}, A_t = \{x_1, \cdots, x_t\}$ and $\tilde F_t = \rbra{\tilde F(x_1), \cdots, \tilde F(x_t)}^{\top}$. 
    \end{definition} 
To establish cumulative regret bounds for GP optimization, a quantity to define is the maximum information gain $\gamma_{t}$: 
\begin{definition}[Maximum Information Gain]
    Define the maximum information gain over $t$ rounds as
    \begin{equation}
     \gamma_t:=\max _{A \subset \mathbb{R}^{d}: |A|=t} \operatorname{I} \left(\mathbf{y}_A ; \mathbf{f^{\star}}_A\right),
    \end{equation}
    where $\operatorname{I} \left(\mathbf{y}_A ; \mathbf{f^{\star}}_A\right)$ denotes the mutual information between $\mathbf{f^{\star}}_A:= [f^{\star}(\phi(x))]_{x \in A}$ and $\mathbf{y}_A: = \mathbf{f^{\star}}_A + \boldsymbol{\epsilon}_A$, where $\boldsymbol{\epsilon}_A \sim \mathcal{N}\left(0, \lambda I_{|A|}\right)$. 
    Thus $\operatorname{I} \left( \mathbf{y}_A ; \mathbf{f^{\star}}_A \right)=\frac{1}{2} \log \operatorname{det} \rbra{ I + \lambda^{-1} K_A}$, where $K_A=\left[k\left(\phi(x), \phi(x^{\prime})\right)\right]_{x, x^{\prime} \in A}$. 
\end{definition}

\subsubsection{Rarely Switching in $f_t$}
\label{subsec:rare_swt}
Posterior sampling is a natural choice to update $f_t$, that is sample $f_t \sim \mathcal{G P}(m_t, k_t)$ at each step $t$. For the sake of a tighter regret bound, we analyse the version of Alg.\ref{alg:meta} where let $V_t= K_{t-1} + \lambda I$,
\begin{equation}
        \label{ft:rare_switch}
        \left\{
            \begin{array}{cc}
            f_t \sim \mathcal{G P}(m_t, k_t), &\text{if } \operatorname{det}(V_t) / \operatorname{det}(V_{pre(t)}) > 2;\\
            f_t \leftarrow f_{t-1}, &\text{otherwise.}
            \end{array}
        \right.
       \end{equation}
with $pre(t)$ denotes the last time step provious to $t$ when $f_t$ is sampled from the posterior.
Here, we keep track of the growth in $\operatorname{det}\rbra{K_{t} + \lambda I}$ and keep $f_t$ the same unless the determinant is doubled. Rarely switching is a common technique in online learning \cite{abbasi2011improved, lansdell2019rarely} to save computation and it helps stabilize Alg.\ref{alg:meta} by reducing the fluctuation in $f_t$ over time. By using this adaptive update schedule, we will show in Theorem \ref{thm:main} that a total of $O \rbra{\gamma_T}$ updates is sufficient without compromising the efficiency in learning $f^{\star}$.
\section{Proof of Themorem \ref{thm:main}}
\label{appx_proof}

Recall the definition of Bayesian regret and write out an equivalent expression with embedding $\phi(x)$:
\begin{align}
\label{equ:rgt}
        \operatorname{BayesRGT}(T) &= \mathbb{E}\left[\sum_{t=1}^T \rbra{ \max_{x \in \bar{\mathcal{X}} } F(x) - F(x_{t})} \right] \notag \\
        & \leq \mathbb{E}\left[\sum_{t=1}^T \rbra{ \max_{\phi(x) \in \mathcal{D} } f^{\star}(\phi(x)) - f^{\star}(\phi(x_t))} \right],
\end{align}
as well as Condition \ref{cond:local_argmax} characterizing the max-fitness series $\{x_t\}^T_{t=1}$ to study:
\begin{align}
\label{equ:ft}
    F_t( x_{t} ) &\geq \max \{ F_t(x) \mid \|\phi(x)-\phi(x_{t-1})\| \leq r\} \notag \\
    \Leftrightarrow f_t(\phi( x_{t}) ) &\geq \max \{ f_t(\phi(x)) \mid \|\phi(x)-\phi(x_{t-1})\| \leq r\}.
\end{align}
From the RHS of (\ref{equ:rgt}) and (\ref{equ:ft}), the embedding vectors $\{\phi(x_t)\}^T_{t=1}$ fully determine the regret and identify the max-fitness sequence  $\{x_t\}^T_{t=1}$. Therefore, the proof of Theorem \ref{thm:main} can be conducted on the embedding level, that is, essentially it suffices to prove Alg.\ref{alg:evo_ker_mnt_TS}(a revision of Alg.\ref{alg:meta}) satisfies
\begin{align}
        \operatorname{BayesRGT}(T) &\leq \mathbb{E}\left[\sum_{t=1}^T \rbra{ \max_{x \in \mathcal{D} } f^{\star}(x) - f^{\star}(x_t)} \right] \notag \\
        &= O \rbra{  \beta_{T}\sqrt{\lambda T \gamma_{T}} + B \gamma_{T} \rbra{1 + \frac{4 L_{\phi}^2 N^2}{r^2}}},
\end{align}
where $\beta_T = O\left( \mathbb{E} \sbra{\|f^{\star}\|_{k}}+\sqrt{d \ln T} + \sqrt{\gamma_{T-1}}\right)$ when the noise level $\lambda = 1 + \frac{1}{T}$.

\paragraph{Important Remark on Notations and Alg.\ref{alg:evo_ker_mnt_TS}.} 
\begin{itemize}
    \item Throughout the proof, we abuse the notation $x$ for an arbitrary sequence in $\mathcal{X}$ to represent its known $d$-dim embedding vector $\phi(x)$ unless clarified by text. Alg.\ref{alg:evo_ker_mnt_TS} rewrites Alg.\ref{alg:meta} on the embedding level (translating each step of Alg.\ref{alg:meta} into what happens to the embedding vectors) to reflect our reduction in analysis.
    \item Alg.\ref{alg:evo_ker_mnt_TS} integrates the pseudo-code for updating $f_t$ corresponding to (\ref{ft:rare_switch}) (Line 4-9) and Condition \ref{cond:local_argmax} by rewriting it as (\ref{equ:cond}) (Line 11). Recall $K_t$ is the kernel matrix, $\lambda$ is the std. of noise, $m_t$ and $k_t$ define the posterior GP, and $u(x_t)$ ($x_t$: embedding vector) corresponds to the noisy feedback $\tilde F(x_t)$ ($x_t$: original sequence in $\mathcal{X}$).
    \item For technical reason, our analysis applies to a modified version of Thompson Sampling with discretization over $\mathcal{D}$ such that
    \begin{equation}
        |f^{\star}(x) - f^{\star}([x]_t)| \leq \frac{1}{t^2}, \quad \forall x \in D, 
    \end{equation}
    $[x]_t$ is the nearest point to $x$ in $\mathcal{D}_t$. Such discretization can be achieved with size $|\mathcal{D}_t| = (L_{\phi}N \|f^{\star}\|_k d M t^2)^d$, $M = \sup _{x \in D} \sup _{j \in[d]}\left(\left.\frac{\partial^2 k(p, q)}{\partial p_j \partial q_j}\right|_{p=q=x}\right)^{1 / 2}$
    The same discretization is adopted by \cite{chowdhury2017kernelized}. 
    
    And in the tree search step, each expanded node $x$ on the search tree is judged by the value $f_t([x]_t)$ in considering whether to keep it in the active set.
\end{itemize}

\begin{algorithm}[htb!]
        \caption{Alg. \ref{alg:meta} Rewritten in the Embedding Space with Monitored Posterior Updates}
        \label{alg:evo_ker_mnt_TS}
        \begin{algorithmic}[1]
            \STATE \textbf{Input:}
            total rounds $T$.
        
            \STATE \textbf{Initialize:} $ x_0 \in \mathcal{D}$, $ k=0, V_0 = \textbf{0}$ and  dataset $S_0 \leftarrow \emptyset$.
            
           \FOR{ $t=1,\ldots , T$} 
                \STATE \textbf{Kernel matrix  monitoring:}
                    \begin{equation*}
                        V = K_{t-1} + \lambda I.
                    \end{equation*}

                \IF{$\operatorname{det}(V) / \operatorname{det}(V_{k}) > 2 $}
                    \STATE \textbf{Update:} $k \leftarrow k+1, V_{k} \leftarrow V.$
                    \STATE \textbf{Discretization}: construct a discretization $D_t \subset D$.
                    \STATE  \textbf{Thompson Sampling:}  For all $x \in D_t$, sample $f_t(x)$ from
                     \begin{equation*}
                         \mathcal{G P}(m_t, k_t).
                     \end{equation*}
                \ELSE 
                \STATE $f_t(\cdot) \leftarrow f_{t-1}(\cdot)$, $\forall x \in D_{t-1}$.
                \ENDIF

                \STATE \textbf{Tree Search:} expand the search tree and select $x$ by value $f_t([x]_t)$, the filtered-out iterate $x_t$ satisfies
                \begin{equation}
                \label{equ:cond}
                    f_t( x_{t} ) \geq \max \{ f_t(x) \mid x: \|x- x_{t-1}\| \leq r\}.
                \end{equation}
                
                \STATE \textbf{Data collection:}
                $S_t \leftarrow S_{t-1} \cup \{(x_{t}, u(x_{t}))\}.$
            \ENDFOR
        \end{algorithmic}
        \end{algorithm}

\subsection{Tree Search and Proximal Iteration}
Given the previous iterate $x_{t-1}$ and an estimate $f_t$ of the fitness, (\ref{equ:cond}) ensures $x_{t}$ is at least as good as the optimal solution of the locally constrained optimization:
\begin{align}
\label{constrained_opt}
    \nonumber &\max_{x} f_t(x),\\
    s.t. \quad & \|x - x_{t-1}\|_2 \leq r.
\end{align}

Consider the proximal point of $x_{t-1}$ under $f_t$ with step size $\alpha$: 
\begin{equation}
\label{proximal}
    \tilde x_t = \mathcal{P}_{\alpha, f_t} (x_{t-1}):= \arg\max_{x \in \mathcal{D}}{f_t(x) - \frac{1}{2 \alpha} \|x - x_{t-1}\|^2}.
\end{equation}
Note that the expression of (\ref{proximal}) differs from the classic definition of proximal point by the constraint on $\mathcal{D}$. However, this is a minor difference as the objective ${f_t(x) - \frac{1}{2 \alpha} \|x - x_{t-1}\|^2}$ is strongly concave on $\mathcal{D}$
and $\tilde x_t$ is in the interior of $\mathcal{D}$ under Assumption \ref{asmp:ft_cvx}. The optimal condition of (\ref{proximal}) ensures
\begin{equation}
    \tilde x_{t} - x_{t-1} \in \alpha \cdot \nabla f_t(\tilde x_t).
\end{equation}

By setting a proper $\alpha$, we are able to show $x_t$ has more improvement in $f_t$ value than $\tilde x_t$.
\begin{proposition}
\label{prop:tree_proximal}
When $\alpha = \frac{r^2}{4B}$, $f_t(x_t) \geq f_t(\tilde x_t)$.
\end{proposition}

\begin{proof}
It's obvious that $f_t(x_t) \geq f_t(\tilde x_t)$ holds when $\tilde x_t \in \mathcal{B} \rbra{x_{t-1}, r}:= \{x: \|x - x_{t-1}\|_2 \leq r\}$. If $\alpha = \frac{r^2}{4B}$, then for $\forall x \not \in \mathcal{B} \rbra{x_{t-1}, r}$, $f_t(x) - \frac{1}{2 \alpha} \|x - x_{t-1}\|^2 \leq f_t(x) - \frac{r^2}{2 \alpha} = f_t(x) - 2B \leq -B \leq f_t(x_{t-1})$, thus $\tilde x_t \in \mathcal{B} \rbra{x_{t-1}, r}$.
\end{proof}

\subsection{Regret Decomposition}
Recall from the preliminaries of GP, $\{f_t\}$ is a sequence of online learning models of the ground truth $f^{\star}$ s.t. 
\begin{equation*}
    \mathbb{E} [\max_{x \in \mathcal{D}} f_t (x)] = \mathbb{E} [\max_{x \in \mathcal{D}} f^{\star} (x)],
\end{equation*}
then the total regret of $\{x_t\}$ is decomposed as
\begin{align*}
    &\mathbb{E} \left[\sum_{t=1}^{T} \max_{x \in \mathcal{D}}  f^{\star} (x) - f^{\star} (x_t) \right]\\
    =& \mathbb{E} \left[\sum_{t=1}^{T} \max_{x \in \mathcal{D}}  f_t (x) - f_t (x_t) \right] + \mathbb{E} \left[\sum_{t=1}^{T} f_t (x_t) - f^{\star} (x_t) \right].
\end{align*}

\paragraph{Re-indexing $\{x_{t}\}_{t=1}^{T}$ (Phase Index).} Suppose $f_t$ is updated at time steps $\{t_k\}_{k=0}^{K}$, where $t_1 = 1, t_2 = 1 + n_1, \cdots, t_K= 1+\sum_{k=1}^{K-1} n_k, T = \sum_{k=1}^{K} n_k$ with $n_k, K\in \mathbb{N}^{\star}$. It indicates that the whole timeline $t \in [T]$ is partitioned into $K$ phases of length $\{n_k\}_{k=1}^{K}$. The $k$-th phase is governed by one unchanging function $f_{t_k}$, i.e. $\forall t \in [T], k \in [K]$ s.t. $t_k \leq t < t_{k+1}$, we have $f_t = f_{t_k}$. To reflect which phase $x_t$ is in, we assign the ``phase index" $x_{k,i}, i \in [n_k]$ to $x_t, t \in [t_k, t_{k+1})$.

\subsection{Regression Error $\mathbb{E} \left[\sum_{t=1}^{T} f_t (x_t) - f^{\star} (x_t) \right]$}
\begin{lemma}
\label{lmm:rare_switch}
    Under Assumption \ref{asmp:gp_f} and \ref{asmp:noisy_fdb}, the total number of updates satifies $K \leq O\rbra{\gamma_T}$ and
    \begin{equation*}
        \mathbb{E} \sbra{\sum_{t=1}^{T} f_t(x_t) - f^{\star}(x_t)} = O \rbra{\beta_{T} \sqrt{\lambda T \gamma_{T}}},
    \end{equation*}
    where $\beta_T = O\left( \mathbb{E} \sbra{\|f^{\star}\|_{k}}+\sqrt{d \ln T} + \sqrt{\gamma_{T-1}}\right)$ and $\gamma_T$ is the maximum information gain.
\end{lemma}

\subsection{Optimization Error $\mathbb{E} \left[\sum_{t=1}^{T} \max_{x \in \mathcal{D}}  f_t (x) - f_t (x_t) \right]$}

To stabilize evolution, we keep using a same model $f_t$ throughout multiple rounds. Re-indexing $\{x_t\}$ to reflect the update of $f_t$ , we have
\begin{align}
    & \mathbb{E} \left[\sum_{t=1}^{T} \max_{x \in \mathcal{D}}  f_t (x) - f_t (x_t) \right] \notag \\=& \mathbb{E} \left[\sum_{k=1}^{K} \sum_{i=1}^{n_k} \max_{x \in \mathcal{D}}  f_{t_k} (x) - f_{t_k} (x_{k,i}) \right].
\end{align}
Here $\{x_{k,i}\}_{i=1}^{n_k}$ is $n_k$-long segment of the whole trajectory seeking to maximize $f_{t_k}$ in an evolutionary fashion. As shown in Proposition \ref{prop:tree_proximal}, for every iterate $x_{t}$, there exists a proximal point iterate $\tilde x_{t}$ whose $f_t$ value lower bounds $f_t(x_t)$, i.e. $\forall k \in [K]$,
\begin{align*}
    \tilde x_{k,1} := x_{k,1}, \quad &\tilde x_{k,i+1} := \mathcal{P}_{\alpha, f_{t_k}} (x_{k,i}), \forall i \in [n_k-1];\\
    f_{t_k} (\tilde x_{k,i}) &\leq f_{t_k} (x_{k,i}).
\end{align*}

\begin{lemma}
\label{lmm:prox_opt}
    Under Assumption \ref{asmp:ft_cvx} and Condition \ref{cond:local_argmax} and with $\alpha = \frac{r^2}{4B}$, 
    \begin{align}
        \sum_{k=1}^{K} \sum_{i=1}^{n_k} \max_{x \in \mathcal{D}}  f_{t_k} (x) - f_{t_k} (\tilde x_{k,i}) &\leq O \rbra {K \rbra{B+ \frac{R^2}{\alpha}}} \notag \\
        &= O \rbra {K B \rbra{1+ 4 \frac{R^2}{r^2}}}.
    \end{align}
    where the RHS also upper bounds $\mathbb{E} \left[\sum_{t=1}^{T} \max_{x \in \mathcal{D}}  f_t (x) - f_t (x_t) \right]$ and by the assumption that the embedding map $\phi$ is Lipschitz, $R = LN$.
\end{lemma}

\subsection{Final Regret Bound}
Combining Lemma \ref{lmm:rare_switch} and \ref{lmm:prox_opt}, the proof of Theorem \ref{thm:main} completes.

\section{Omitted Proofs in Appendix \ref{appx_proof}}
\subsection{Proof of Lemma \ref{lmm:rare_switch}}
\begin{proof}
We bring the discretization $\mathcal{D}_t$ in to decompose $\mathbb{E} \left[\sum_{t=1}^{T} f_t (x_t) - f^{\star} (x_t) \right]$. Given $[x_t]_t \in \mathcal{D}_t$ is the nearest point to $x_t$, then
\begin{align*}
    &\mathbb{E} \left[\sum_{t=1}^{T} f_t (x_t) - f^{\star} (x_t) \right] \\
    \leq &\mathbb{E} \left[\sum_{t=1}^{T} f_t (x_t) - f_t ([x_t]_t) \right] + \mathbb{E} \left[\sum_{t=1}^{T} f_t ([x_t]_t) - f^{\star} ([x_t]_t) \right] \\
    &+ \mathbb{E} \left[\sum_{t=1}^{T} f^{\star} ([x_t]_t) - f^{\star} (x_t) \right] \\
    \leq &\sum_{t=1}^{T} \mathbb{E} \left[ f_t ([x_t]_t) - f^{\star} ([x_t]_t) \right] + O \left(\sum_{t=1}^{T} \frac{1}{t^2}\right),
\end{align*}
where we utilize the properties that $|f^{\star}(x) - f^{\star}([x]_t)| \leq \frac{1}{t^2}$ and $f_t$ is $L_f$ Lipschitz continuous in $\mathcal{D}$, thus there exists some constant $C_0$ s.t. $|f^{\star}(x) - f^{\star}([x]_t)| \leq \frac{C_0}{t^2}$.

By our updating rule for the posterior
\begin{align}
    2^{K} \leq & \frac{\operatorname{det}\rbra{\lambda I + K_{t_1 - 1}}}{\operatorname{det}\rbra{\lambda I}} \cdot \frac{\operatorname{det}\rbra{\lambda I + K_{t_2 - 1}}}{\operatorname{det}\rbra{\lambda I + K_{t_1 - 1}}} \cdots \notag \\
    & \frac{\operatorname{det}\rbra{\lambda I + K_{t_K - 1}}}{\operatorname{det}\rbra{\lambda I + K_{t_{K-1} - 1}}}
    \leq \operatorname{det}\rbra{ I + \lambda^{-1} K_{T}}.
\end{align}
Therefore,
\begin{equation}
    K \leq \ln \rbra{\operatorname{det} \rbra{I + \lambda ^{-1} K_{T}}} \leq 2\gamma_T.
\end{equation}

Denote by $\|f^{\star}\|_{k}$ the norm of function $f^{\star}$ in the RKHS associated with kernel $k(\cdot, \cdot)$.
From the Theorem 2 in \cite{chowdhury2017kernelized}, then with probability at least $1-\delta$, the following holds for all $x \in \mathbb{R}^{d}$ and $t \geq 1$ :
\begin{align}
\label{equ_conc_mt_f} 
    & \left|m_{t}(x)-f^{\star}(x)\right| \notag \\
    \leq & \left(\|f^{\star}\|_{k} + \sqrt{\lambda} \sqrt{2 \rbra{\ln \frac{1}{\delta}}  + \ln \rbra{\operatorname{det} \rbra{\lambda I+K_{t-1}}} } \right) \sigma_{t}(x)\\
    \leq & \left(\|f^{\star}\|_{k} + \sqrt{\lambda} \sqrt{2 \rbra{\ln \frac{1}{\delta}}  + (t-1)\ln \lambda + \gamma_{t-1}} \right) \sigma_{t}(x),
\end{align}
where $m_{t}(x), \sigma_{t}^2(x)$ are the mean and variance of the posterior distribution of $f^{\star}(x)$ defined as in (\ref{equ_post_mean}) and (\ref{equ_post_cov}). 

Also, according to the marginal sampling distribution of $f_{t}(x)$, which is $\mathcal{N} \rbra{m_{t}(x), \sigma_{t}(x)}$, with probability $1-\delta$, it holds for all $\{[x_t]_{t}\}_{t=1}^{T}$ (or $\{\{[x_{k,i}]_{t_k}\}_{i=1}^{n_k}\}_{k=1}^{K}$ indexed by phases) that 
\begin{align*}
    &\abs{f_{t_k} \rbra{[x_{k,i}]_{t_k}} - m_{t_k}([x_{k,i}]_{t_k})} \\
    \leq & \sigma_{t_k}\rbra{[x_{k,i}]_{t_k}} \sqrt{\ln \rbra{\frac{T |D_T|}{\delta}}}\\
    \leq &\sigma_{t_k}\rbra{[x_{k,i}]_{t_k}} \left(\sqrt{d \ln \rbra{d L_{\phi}N \|f^{\star}\|_k M T^2}} + \sqrt{\ln \rbra{\frac{T}{\delta}}}\right).
\end{align*}

Let $\beta_T = \|f^{\star}\|_{k} + \sqrt{\lambda} \sqrt{2 \rbra{\ln \frac{1}{\delta}}  + (T-1)\ln \lambda + \gamma_{T-1}} + \sqrt{d \ln \rbra{d L_{\phi}N \|f^{\star}\|_k M T^2}} + \sqrt{\ln \rbra{\frac{T}{\delta}}}$, then
\begin{align*}
    &\sum_{k=1}^{K} \sum_{i=1}^{n_k} \abs{f_{t_k}(x_{k, i}) - f^{\star}(x_{k, i})} \\
    = & O\left(\sum_{t=1}^{T}  f_t ([x_t]_t) - f^{\star} ([x_t]_t) \right)\\
    \leq &  \beta_{T} \sum_{k=1}^{K} \sum_{i=1}^{n_k} \sigma_{t_k}\rbra{x_{k,i}}\\
    \leq & \beta_{T} \sqrt{T \lambda \sum_{k=1}^{K} \sum_{i=1}^{n_k}  \lambda^{-1} \sigma^2_{t_k}\rbra{x_{k,i}}}\\
    \leq & \beta_{T} \sqrt{4 T \lambda \sum_{k=1}^{K} \sum_{i=1}^{n_k} \frac{1}{2}\ln \rbra{1 + \lambda^{-1} \sigma^2_{t_k}\rbra{x_{k,i}}}} .
\end{align*}

According to Lemma 3 in \cite{chowdhury2017kernelized}, we have
\begin{equation}
    \frac{1}{2} \sum_{t=1}^T \ln \left(1+\lambda^{-1} \sigma_{t}^2\left(x_t\right)\right) \leq  \gamma_t.
\end{equation}

Suppose $x_{k,i}$ corresponds to $x_{t}$ under  the phase indexing, consider the following ratio
\begin{equation*}
    \frac{\sigma^2_{t_k}\rbra{x_{k,i}}}{\sigma^2_{t}\left(x_t\right)},
\end{equation*}
where $t_k \leq t-1$ is the time step where the posterior Gaussian process was updated for the $k$-th time.
Recall from (\ref{equ_post_cov}),
\begin{equation*}
    \sigma_t^2(x) = k_t \left(x, x\right) =k\left(x, x\right)-\mathbf{k_{t-1}}(x)^{\top} \left(K_{t-1}+ \lambda I \right)^{-1}\mathbf{k_{t-1}}(x),
\end{equation*}
where $k(x, x^{\prime}) = \psi(x)^{\top} \psi(x^{\prime})$ with $\psi(x)$ being the possibly infinite dimensional feature of $x$. Hence by defining the $t \times \infty$ matrix $\Psi_t = \{\psi(x_1), \cdots, \psi(x_t)\}^{\top}$, we have $K_{t-1} = \Psi_{t-1} \Psi_{t-1}^{\top}$ and
\begin{align*}
    \sigma_t^2(x) &= \psi(x)^{\top} \psi(x) -\psi(x)^{\top} \Psi_{t-1}^{\top}\left(K_{t-1}+ \lambda I \right)^{-1} \Psi_{t-1} \psi(x)\\
    &= \psi(x)^{\top} \rbra{I_{t-1} - \Psi_{t-1}^{\top}\left(K_{t-1}+ \lambda I \right)^{-1} \Psi_{t-1}} \psi(x)\\
    &= \lambda \psi(x)^T\left( \Psi_{t-1}^T \Psi_{t-1}+ \lambda I\right)^{-1} \psi(x).
\end{align*}

Therefore,
\begin{align}
     \frac{\sigma^2_{t_k}\rbra{x_{k,i}}}{\sigma^2_{t}\left(x_t\right)} &\leq  \frac{\operatorname{det} \rbra{\lambda^{-1} \Psi_{t-1}^T \Psi_{t-1}+  I}}{\operatorname{det} \rbra{\lambda^{-1} \Psi_{t_k-1}^T \Psi_{t_k-1}+ I}} \notag\\ &= \frac{\operatorname{det} \rbra{\lambda^{-1} K_{t-1}+  I}}{\operatorname{det} \rbra{\lambda^{-1} K_{t_k-1}+ I}} \notag \\ &\leq \frac{\operatorname{det} \rbra{\lambda^{-1} K_{t_{k+1}-1}+  I}}{\operatorname{det} \rbra{\lambda^{-1} K_{t_k-1}+ I}} \leq 2,
\end{align}
where we used $t_k \leq t \leq t_{k+1}$ and we end up with
\begin{equation}
\label{equ_snb}
     \sum_{k=1}^{K} \sum_{i=1}^{n_k} \abs{f_{t_k}(x_{k, i}) - f^{\star}(x_{k, i})} \leq O \rbra{\beta_{T} \sqrt{\lambda T \gamma_{T}}}.
\end{equation}
holds for any realization of $f^{\star}$ with probability $\frac{1}{\delta}$ when $\beta_T = \|f^{\star}\|_{k} + \sqrt{\lambda} \sqrt{2 \rbra{\ln \frac{1}{\delta}}  + (T-1)\ln \lambda + \gamma_{T-1}} +  \sqrt{\ln \rbra{\frac{T}{\delta}}}$. 

Taking $\delta = \frac{1}{T}$ and $\lambda = 1+ \frac{1}{T}$ and taking expectation of \ref{equ_snb} on both side, we have
\begin{equation*}
    \mathbb{E} \sbra{\sum_{k=1}^{K} \sum_{i=1}^{n_k} \abs{f_{t_k}(x_{k, i}) - f^{\star}(x_{k, i})}} \leq O \rbra{\beta_{T} \sqrt{\lambda T \gamma_{T}}}
\end{equation*}
with $\beta_T = \mathbb{E} \sbra{\|f^{\star}\|_{k}} + \sqrt{\lambda} \sqrt{2 \ln T + (T-1)\ln \lambda + \gamma_{T-1}} +  \sqrt{2 \ln T} + \sqrt{d \ln \rbra{d L_{\phi}N \|f^{\star}\|_k M T^2}}\leq \mathbb{E} \sbra{\|f^{\star}\|_{k}} + \sqrt{d} \cdot \mathbb{E} \sbra{\sqrt{\ln \|f^{\star}\|_k}}+ \mathbb{E} \sbra{\sqrt{d \ln (dL_{\phi}NM)}} + \sqrt{d \ln T} + \sqrt{\gamma_{T-1}} = O\left( \mathbb{E} \sbra{\|f^{\star}\|_{k}}+\sqrt{d \ln T} + \sqrt{\gamma_{T-1}}\right)$ .

\end{proof}

\subsection{Proof of Lemma \ref{lmm:prox_opt}}
\begin{lemma}
\label{lmm:prox}
    For a concave $f$ with $z^{\star}$ being one of its maxima, a sequence of iterates ${z_t}$ following the proximal point update 
    \begin{equation}
        z_{t+1} = \mathcal{P}_{\alpha, f} (z_{t})
    \end{equation}
    satisfies
    \begin{equation}
        \sum_{t} f (z^{\star}) - f (z_t)  \leq \frac{1}{2 \alpha} \| z_1 - z^{\star}\|^2.
    \end{equation} 
\end{lemma}

\begin{proof}
Given $z_{t+1} = \mathcal{P}_{\alpha, f} (z_{t})$, we have
\begin{equation}
    z_{t+1} - z_{t} \in \alpha \cdot \partial f(z_{t+1}).
\end{equation}
For any $z$, by the concavity of $f$
\begin{align*}
    f(z) -  f(z_{t+1}) &\leq  \langle \partial f(z_{t+1}), z - z_{t+1} \rangle\\
    &= \frac{1}{\alpha} \langle z_{t+1} - z_{t}, z - z_{t+1} \rangle\\
    &= \frac{1}{2 \alpha} \left[  \|z - z_{t}\|^2 - \| z_{t+1} - z_{t}\|^2 - \|z - z_{t+1}\|^2  \right]\\
    &\leq \frac{1}{2 \alpha} \left[  \|z - z_{t}\|^2 - \|z - z_{t+1}\|^2  \right].
\end{align*}
 Taking $z$ as $z^{\star}$, we have
 \begin{equation*}
     f(z^{\star}) -  f(z_{t+1}) \leq \frac{1}{2 \alpha} \left[  \|z^{\star} - z_{t}\|^2 - \|z^{\star} - z_{t+1}\|^2  \right].
 \end{equation*}
\end{proof}

By Lemma \ref{lmm:prox}, for $\forall k \in [K]$, under Assumption \ref{asmp:ft_cvx}, suppose $x_k^{\star} \in \arg\max_{x \in \mathcal{D}}  f_{t_k} (x)$, we have
\begin{align*}
    &\sum_{i=1}^{n_k} \max_{x \in \mathcal{D}}  f_{t_k} (x) - f_{t_k} (\tilde x_{k,i}) \\
    &\leq \max_{x \in \mathcal{D}}  f_{t_k} (x) - f_{t_k} (\tilde x_{k,1}) + \frac{1}{2 \alpha} \|x_k^{\star} - \tilde x_{k,1}\|^2 \\
    &\leq 2B + \frac{2 R^2}{\alpha}.
\end{align*}
\end{document}